\documentclass{article}
\usepackage[utf8]{inputenc}

\usepackage{hyperref}
\newcommand{\footremember}[2]{%
    \footnote{#2}
    \newcounter{#1}
    \setcounter{#1}{\value{footnote}}%
}
\newcommand{\footrecall}[1]{%
    \footnotemark[\value{#1}]%
} 

\usepackage{microtype}
\usepackage{graphicx}
\usepackage{subfigure}
\usepackage{booktabs}
\usepackage{float}
\usepackage{booktabs}
\usepackage{hyperref}
\usepackage{geometry}
\usepackage{algorithm2e}
\usepackage{latexsym}
\usepackage{amscd}
\usepackage{amsmath}
\usepackage{amssymb}
\usepackage{amsthm}
\usepackage{mathrsfs}
\usepackage{algorithmic}

\usepackage{cleveref}
\theoremstyle{plain}
\newtheorem{theorem}{Theorem}[section]
\crefname{theorem}{Theorem}{Theorems}
\newtheorem{corollary}[theorem]{Corollary}
\crefname{corollary}{Corollary}{Corollaries}

\crefname{lemma}{Lemma}{Lemmas}
\newtheorem{definition-lemma}[theorem]{Definition-Lemma}

\crefname{claim}{Claim}{Claims}

\newtheorem{proposition}[theorem]{Proposition}
\crefname{proposition}{Proposition}{Propositions}

\crefname{conjecture}{Conjecture}{Conjectures}

\crefname{assumption}{Assumption}{Assumptions}

\theoremstyle{definition}
\newtheorem{definition}[theorem]{Definition}
\crefname{definition}{Definition}{Definitions}
\newtheorem{remark}[theorem]{Remark}
\crefname{remark}{Remark}{Remarks}

\crefname{example}{Example}{Examples}

\crefname{notation}{Notation}{Notations}

\crefname{section}{Section}{Sections}
\crefname{subsection}{Subsection}{Subsections}
\crefname{equation}{Equation}{Equations}
\crefname{figure}{Figure}{Figures}
\crefname{table}{Table}{Tables}
\crefname{algorithm}{Algorithm}{Algorithms}

\usepackage{autonum}
\usepackage{comment}

\newcommand{\kmax}{\mathop{k\operatorname{-max}}\limits}

\newcommand{\bR}{\ensuremath{\mathbb{R}}}
\newcommand{\cD}{\ensuremath{\mathcal{D}}}
\newcommand{\cM}{\ensuremath{\mathcal{M}}}
\newcommand{\reg}{\mathrm{reg}}
\newcommand{\ess}{\mathrm{ess}}

\title{Optimizing persistent homology based functions}
\author{Mathieu Carri\`ere\footremember{inriaSAM}{DataShape, Inria Sophia-Antipolis, Biot, France} \and Fr\'ed\'eric Chazal\footremember{inriaSAC}{DataShape, Inria Saclay, Palaiseau, France} \and Marc Glisse\footrecall{inriaSAC} \and Yuichi Ike\footnote{Fujitsu Laboratories, AI Lab, Tokyo, Japan} \and Hariprasad Kannan}
\date{\today}

\begin{document}

\maketitle

\begin{abstract}
Solving optimization tasks based on functions and losses with a topological flavor is a very active,
growing field of research in data science and Topological Data Analysis, 
with applications in non-convex optimization, statistics and machine learning. 
However, the approaches proposed in the literature are usually
anchored to a specific application and/or topological construction, and do not come with theoretical guarantees. 
To address this issue, we study the differentiability of a general map associated with the most common topological construction, that is, the persistence map. 
Building on real analytic geometry arguments, we propose a general framework that allows us to define and compute gradients for persistence-based functions in a very simple way. 
We also provide a simple, explicit and sufficient condition for convergence of stochastic subgradient methods for such functions.
This result encompasses all the constructions and applications of topological optimization in the literature.
Finally, we provide associated code, that is easy to handle and to mix with other non-topological methods and constraints,
as well as some experiments showcasing the versatility of our approach.
\end{abstract}

\section{Introduction}\label{sec:intro}

Persistent homology is a central tool in Topological Data Analysis that allows to efficiently infer relevant topological features of complex data in a descriptor called persistence diagram. 
It found many applications in Machine Learning (ML) where it initially played the role of a feature engineering tool through the direct use of persistence diagrams, or through dedicated ML architectures to handle them---see, e.g., \cite{hofer2017deep,umeda2017time, carriere2020perslay,dindin2020topological, kim2020efficient}.
For the last few years, a growing number of works have been successfully using persistence theory in different ways in order to, for instance, better understand, design and improve neural network architectures---see, e.g., \cite{rieck2018neural,moor2019topological,carlsson2020topological,gabrielsson2019exposition}---or design regularization and loss functions incorporating topological terms and penalties for various ML tasks---see, e.g., \cite{chen2019topological, hofer2019connectivity, clough2020topological}.
These new use cases of persistence generally involve minimizing functions that depend on persistence diagrams. Such functions are in general highly non-convex and not differentiable, and thus their theoretical and practical minimization can be difficult. 
In some specific cases, persistence-based functions can be designed to be differentiable and/or some effort have to be  made to compute their gradient, so that standard gradient descent techniques can be used to minimize them---see e.g., \cite{wangtopogan,poulenard2018topological, bruel2020topology}. 
In the general case, recent attempts have been made to better understand their differential structures~\cite{leygonie2019framework}.
Moreover, building on powerful tools provided by software libraries such as PyTorch or TensorFlow, practical methods allowing to encode and optimize a large family of persistence-based functions have been proposed and experimented~\cite{gabrielsson2020topology, solomon2020fast}. 
However, in all these cases, the algorithms used to minimize these functions do not come with theoretical guarantees of convergence to a global or local minimum. 

%\vspace{-10mm}
\paragraph{Contributions and organization of the article.}
The aim of this article is to provide a general framework that includes almost all persistence-based functions from the literature, and for which stochastic subgradient descent algorithms are easy to implement and come with convergence guarantees. 

More precisely, we first observe that the persistence map, converting a filtration over a given simplicial complex
\footnote{The presentation is restricted to simplicial complexes for simplicity, but this generalizes to other complexes as well. We present an example with cubical complexes in Appendix.}
into a persistence diagram, can be thought of as a map between Euclidean spaces (\cref{sec:filt}). 
This observation allows us to prove that the persistence map is semi-algebraic and, using classical arguments from o-minimal geometry, to study the differentiability of the persistence of parametrized families of filtrations (\cref{sec:ominim}). 
Then, building on the recent work \cite{subgradient}, we consider the minimization problem of persistence-based functions and show that under mild assumptions, stochastic subgradient descent algorithms applied to such functions converge almost surely to a critical point (\cref{sec:funcs}). 
We also provide a simple corresponding Python implementation\footnote{\url{https://github.com/MathieuCarriere/difftda.git}} for minimizing functions of persistence, and we illustrate it with several examples from the literature (\cref{sec:implementation}).

\section{Filtrations and persistence diagrams} \label{sec:filt}

In this section, we show that the persistence map is nothing but a permutation of the coordinates of a vector containing the filtration values.

\subsection{Simplicial complexes and filtrations}

Recall that given a set $V$, a (finite) \emph{simplicial complex} $K$ is a collection of finite subsets of $V$ that satisfies
(1) $\{v \} \in K$ for any $v \in V$, and
(2) if $\sigma \in K$ and $\tau \subseteq \sigma$ then $\tau \in K$.
An element $\sigma \in K$ with $|\sigma|=k+1$ is called a $k$-simplex. 

Given a simplicial complex $K$ and a subset $R$ of $\bR$, a \emph{filtration} of $K$ is an increasing sequence $(K_r)_{r \in R}$ of subcomplexes of $K$ with respect to the inclusion, i.e., $K_r \subseteq K_s$ for any $r \leq s$, and such that $\bigcup_{r \in R} K_r = K$.

To each simplex $\sigma \in K$, one can associate its \emph{filtering index} $\Phi_\sigma = \inf \{ r \in R: \sigma \in K_r \}$. 
Thus, when $K$ is finite, a filtration of $K$ can be conveniently represented as a filtering function $\Phi \colon K \to \bR$.
Equivalently, it can be represented as a $|K|$-dimensional vector $\Phi  = (\Phi_\sigma)_{\sigma \in K}$ in $\mathbb{R}^{|K|}$ whose coordinates are the indices of the simplices of $K$ and that satisfies the following condition: 
if $\sigma, \tau \in K$ and $\tau \subseteq \sigma$, then $\Phi_\tau \le \Phi_\sigma$.
As a consequence, if the vectorized filtration $\Phi$ depends on a parameter, the corresponding family of filtrations can be represented as a map from the space of parameters to $\bR^{|K|}$ in the following way.

\begin{definition} \label{def:parametrized_filtrations}
    Let $K$ be a simplicial complex and $A$ a set. 
    A map $\Phi \colon A \to \bR^{|K|}$ is said to be a \emph{parametrized family of filtrations} if for any $x \in A$ and $\sigma, \tau \in K$ with $\tau \subseteq \sigma$, one has $\Phi_\tau(x) \le \Phi_\sigma(x)$. 
\end{definition}

\subsection{Persistence computation from filtrations}
\label{compute-pers}

We briefly recall how the computation of the persistence diagram of a filtered simplicial complex decomposes into: (i)~a purely combinatorial part only relying on the order on the simplices induced by the filtration, and (ii)~a part relying on the filtration values. 
A detailed introduction to persistent homology and its computation can be found in, e.g., \cite{edelsbrunner2010computational,boissonnat2018geometric}.

Let $K$ be a simplicial complex endowed with a filtration and corresponding filtering function $\Phi \in \bR^{|K|}$, where $|K|$ is the number of non-empty simplices of $K$.
%The computation of the persistence diagram can be split into two steps.

\paragraph{First part: combinatorial part (persistence pairs).}
The filtering function $\Phi$ induces a total preorder on the elements of $K$ as follows: 
$\tau \preceq \sigma$ if $\Phi_\tau \leq \Phi_\sigma$. 
This preorder can be refined into a total order by breaking ties in some fairly arbitrary way, as long as it is consistent with the face relation, i.e., if $\tau\subseteq\sigma$, then $\tau \preceq \sigma$. 
One way to break ties is to sort simplices that have the same filtration value by dimension, and then order the ones that are still equivalent according to some arbitrary indexing of the simplices. 
A different way is to index the vertices, represent simplices by their decreasing list of vertices, and sort equivalent simplices using the lexicographic order on those lists. 
In the following, we will assume that the total order is a function of the preorder, in particular it is deterministic and does not depend on the exact values of $\Phi$. 
Note that while different orders may yield different pairings, they all translate to the same persistence diagram in the second part.
The basic algorithm to compute persistence iterates over the ordered set of simplices $\sigma_1 \preceq \dots \preceq \sigma_{|K|}$ according to \cref{alg:persistent-pairs} below---see Section~11.5.2 in \cite{boissonnat2018geometric} for a detailed description of the algorithm.

\begin{algorithm}[ht]
  \caption{Persistence pairs computation (sketch)}
  \label{alg:persistent-pairs}
  \begin{algorithmic}
  	\STATE{\bf Input:} Ordered sequence of simplices $\sigma_1 \preceq \dots \preceq \sigma_{|K|}$
  	\STATE $K_0 = \emptyset$
    \STATE $\mathrm{Pairs}_0 = \mathrm{Pairs}_1 = \dots = \mathrm{Pairs}_{d-1} = \emptyset$
    \FOR{$j=1$ to $|K|$}
    	\STATE $k = \dim \sigma_j$
    	\STATE $K_j = K_{j-1} \cup \sigma_j$
    	\IF{$\sigma_j$ does not create a new $k$-dimensional homology class in $K_j$}
    	    \STATE a $(k-1)$-dimensional homology class created in $K_{l(j)}$ by $\sigma_{l(j)}$ for some $l(j) < j$ becomes homologous to $0$ in $K_j$.
    		%\STATE $l(j) =$ the largest index of the positive $k$-simplices associated to $\partial \sigma^j$;
    		\STATE $\mathrm{Pairs}_{k-1} \gets \mathrm{Pairs}_{k-1} \cup \{ (\sigma_{l(j)}, \sigma_j) \}$;
    	\ENDIF
    \ENDFOR
    \STATE{\bf Output:} Persistence pairs in each dimension $\mathrm{Pairs}_0, \mathrm{Pairs}_1, \dots, \mathrm{Pairs}_{d-1}$
  \end{algorithmic}
\end{algorithm}

Note that for each dimension $k$, some $k$-dimensional simplices may remain unpaired at the end of the algorithm; their number is equal to the $k$-dimensional Betti number of $K$. 

\paragraph{Second part: associated filtration values.}

The persistence diagram of the filter function $\Phi$ is now obtained by associating to each persistent pair $(\sigma_{l(j)}, \sigma_j)$ the point 
$(\Phi_{\sigma_{l(j)}}, \Phi_{\sigma_j})$.
Moreover, to each unpaired simplex $\sigma_l$ is associated the point $(\Phi_{\sigma_l},+\infty)$.  

If $p$ is the number of persistence pairs and $q$ is the number of unpaired simplices, then $|K|=2p+q$ and the persistence diagram $D(\Phi)$ of the filtration $\Phi$ of $K$ is made of $p$ points in $\bR^2$ (counted with multiplicity) and $q$ points (also counted with multiplicity) with infinite second coordinate.
Choosing the lexicographical order on $\bR \times (\bR \cup \{+\infty \})$, the persistence diagram $D(\Phi)$ can be represented as a vector in $\bR^{|K|}$ and the output of the persistence algorithm can be simply seen as a permutation of the coordinates of the input vector $\Phi$. 
Moreover, this permutation only depends on the order on the simplices of $K$ induced by $\Phi$. 

\begin{definition}\label{def:regess}
     The subset of points of a persistence diagram $D$ with finite coordinates (resp.\ infinite second coordinate) is called the \emph{regular part} (resp.\ \emph{essential part}) of $D$ and denoted by $D_{\reg}$ (resp.\ $D_{\ess}$).
\end{definition}

With the notations defined above, $D_{\reg}$ and $D_{\ess}$ can be represented as vectors in $\bR^{2p}$ and $\bR^q$, respectively. 

 Note that, in practice, the above construction is usually done dimension by dimension to get a persistence diagram for each dimension in homology, by restricting to the subset of simplices of dimension $k$ and $k+1$. Without loss of generality, and to avoid unnecessary heavy notation, in the following we consider the whole persistence diagram, made of the union of the persistence diagrams in all dimensions~$k$. 

\begin{comment}
\begin{remark}
    Note that, in practice, the above construction is usually done dimension by dimension, leading to a single persistence diagram for each dimension in homology. 
    Indeed, in order to obtain a persistence diagram in dimension $k$, it suffices to  restrict to the subset of simplices of dimension $k$ and $k+1$. 
    Without loss of generality, and to avoid unnecessary heavy notation, in the remaining of this article, we consider the whole persistence diagram, made of the union of the persistence diagrams in all dimensions~$k$. 
\end{remark}
\end{comment}

\section{Differentiability of functions of persistence}\label{sec:ominim}

o-minimal geometry provides a well-suited setting to describe the parametrized families of filtrations encountered in practice and to exhibit interesting differentiability properties of their composition with the persistence map.

\subsection{Background on o-minimal geometry}

In this section, we recall some elements of o-minimal geometry, which are needed in the next sections---see, e.g., \cite{coste2000introduction} for a more detailed introduction. 

\begin{definition}[o-minimal structure]
    An \emph{o-minimal structure} on the field of real numbers $\mathbb{R}$ is a collection $(S_n)_{n \in \mathbb{N}}$, where each $S_n$ is a set of subsets of $\mathbb{R}^n$ such that:
    \begin{enumerate}
        \itemsep -0.1cm
        \item $S_1$ is exactly the collection of finite unions of points and intervals; 
        \item all algebraic subsets\footnote{Recall that an algebraic set is the $0$-level set of a polynomial.} of $\mathbb{R}^n$ are in $S_n$;
        \item $S_n$ is a Boolean subalgebra of $\bR^n$ for any $n \in \mathbb{N}$; 
        \item if $A \in S_n$ and $B \in S_m$, then $A \times B \in S_{n+m}$;
        \item if $\pi \colon \bR^{n+1} \to \bR^n$ is the linear projection onto the first $n$ coordinates and $A \in S_{n+1}$, then $\pi(A) \in S_n$.
    \end{enumerate}
\end{definition}

An element $A \in S_n$ for some $n \in \mathbb{N}$ is called a \emph{definable set} in the o-minimal structure. 
For a definable set $A \subseteq \bR^n$, a map $f \colon A \to \bR^m$ is said to be \emph{definable} if its graph is a definable set in $\bR^{n+m}$.

Definable sets are stable under various geometric operations. 
The complement, closure and interior of a definable set are definable sets. 
The finite unions and intersections of definable sets are definable. 
The image of a definable set by a definable map is itself definable. 
Sums and products of definable functions as well as compositions of definable functions are definable---see Section~1.3 in \cite{coste2000introduction}. 
In particular, the $\max$ and $\min$ of finite sets of real-valued definable functions are also definable.  
An important property of definable sets and definable maps is that they admit a finite Whitney stratification. 
This implies that (i)~any definable set $A \subseteq \bR^n$ can be decomposed into a finite disjoint union of smooth submanifolds of $\bR^n$ and (ii)~for any definable map $\Phi \colon A \to \bR^m$, $A$ can also be decomposed into a finite union of smooth manifolds such that the restriction of $\Phi$ on each of these manifolds is a differentiable function. 

The simplest example of o-minimal structures is given by the family of semi-algebraic subsets\footnote{It is the family of all finite unions and intersections of level sets and sublevel sets of polynomials \cite{benedetti1991real}.} of $\bR^n \ (n \in \mathbb{N})$.
Although most of the classical parametrized families of filtrations are semi-algebraic, the o-minimal framework actually allows to consider larger families.
In particular, the result of \cite{wilkie1996model} says that the family of images of the sublevel sets of functions in $\bR[x_1,\dots,x_N,\exp(x_1),\dots,\exp(x_N)]$ for some $N \in \mathbb N$ under linear projections is an o-minimal structure, which allows us to mix exponential functions with semi-algebraic functions.

\subsection{Persistence diagrams of definable parametrized families of filtrations}\label{sec:pers}

Let $K$ be a simplicial complex and $\Phi \colon A \to \bR^{|K|}$ be a parametrized family of filtrations that is definable in a given o-minimal structure.  
If for any $x, x' \in A$, the preorders induced by $\Phi(x)$ and $\Phi(x')$ on the simplices of $K$ are the same, i.e., for any $\sigma_1, \sigma_2 \in K$, $ \Phi_{\sigma_1}(x) \leq \Phi_{\sigma_2}(x)$ if and only if $\Phi_{\sigma_1}(x') \leq \Phi_{\sigma_2}(x')$, then the pairs of simplices $(\sigma_{i_1},\sigma_{j_1}), \dots , (\sigma_{i_p},\sigma_{j_p})$, and the unpaired simplices $\sigma_{i_{p+1}},\dots, \sigma_{i_{p+q}}$ that are computed by the persistence \cref{alg:persistent-pairs} are independent of $x$. 
As a consequence, the persistence diagram $D = D(\Phi(x))$ of $\Phi(x)$ is 
\begin{align}  \label{eq:diff_dgm}
    D = \bigcup_{k=1}^p (\Phi_{\sigma_{i_k}}(x),\Phi_{\sigma_{j_k}}(x)) \cup \bigcup_{k=1}^q (\Phi_{\sigma_{i_{p+k}}}(x),+\infty),
\end{align}
where $|K|=2p+q$.

Given the lexicographic order on $\bR \times (\bR \cup \{+\infty \})$, the points of any finite multi-set $D \subseteq \bR \times (\bR \cup \{+\infty \})$ with $p$ points in $\bR^2$ and $q$ points in $\bR \times \{ +\infty \}$ can be ordered in non-decreasing order, and $D$ can be represented as a vector in $\bR^{2p+q}$. 
As a consequence, denoting by $\mathrm{Filt}_K$ the set of vectors in $\bR^{|K|}$ defining a filtration on $K$, the persistence map $\mathrm{Pers} \colon \mathrm{Filt}_K \to \bR^{|K|}$ that assigns to each filtration of $K$ its persistence diagram consists in a
%piecewise constant 
permutation of the coordinates of $\bR^{|K|}$. 
This permutation is constant on the set of filtrations that define the same preorder on the simplices of $K$. 
This leads to the following statement.

\begin{proposition} \label{lemma:semi-alg-persistence}
Given a simplicial complex $K$, the map $\mathrm{Pers} \colon \mathrm{Filt}_K \subseteq \bR^{|K|} \to  \bR^{|K|}$ is semi-algebraic, and thus definable in any o-minimal structure. 
Moreover, there exists a semi-algebraic partition of $\mathrm{Filt}_K$ such that the restriction of $\mathrm{Pers}$ to each element of this partition is a Lipschitz map. 
\end{proposition}

\begin{proof}
See Appendix.
\end{proof}

Since there exists a finite semi-algebraic partition of $\mathrm{Filt}_K$ on which $\mathrm{Pers}$ is a locally constant permutation, the subdifferential (see \cref{sec:funcs} for the definition) of $\mathrm{Pers}$ is well-defined and obvious to compute: each coordinate in the output (i.e., the persistence diagram) is a copy of a coordinate in the input (i.e., the filtration values of the simplices). This implies that every partial derivative is either 1 or 0.
The output can be seen as a reindexing of the input, and this is indeed how we implement it in our code, so that automatic differentiation frameworks (PyTorch, TensorFlow, etc.) can process the function $\mathrm{Pers}$ directly and do not need explicit gradient formulas---see \cref{sec:implementation}.
Note that the subdifferential depends on the arbitrary refinement of the preorder in \cref{compute-pers}.

\begin{corollary}
    Let $K$ be a simplicial complex and $\Phi \colon A \to \bR^{|K|}$ be a definable (in a given o-minimal structure) parametrized family of filtrations.
    The map $\mathrm{Pers} \circ \Phi \colon A \to \bR^{|K|}$ is definable.  
\end{corollary}

Note that according to the remark following \cref{lemma:semi-alg-persistence}, if $\Phi$ is differentiable, the subdifferential of $\mathrm{Pers} \circ \Phi$ can be easily computed in terms of the partial derivatives of $\Phi$ using, for example, \cref{eq:diff_dgm}. 

It also follows from standard finiteness and stratifiability properties of definable sets and maps that $\mathrm{Pers} \circ \Phi$ is differentiable almost everywhere. 
More precisely: 

\begin{proposition} \label{thm:diff_semialg}
    Let $K$ be a simplicial complex and $\Phi \colon A \to \bR^{|K|}$ a definable parametrized family of filtrations, where $\dim A = m$.
    Then there exists a finite definable partition of $A$, $A = S \sqcup O_1 \sqcup \dots \sqcup O_k$ 
    such that $\dim S < \dim A := m$ and, for any $i=1,\dots,k$, $O_i$ is a definable manifold of dimension $m$ and $\mathrm{Pers} \circ \Phi \colon O_i \to \cD$ is differentiable.
\end{proposition}

\begin{comment}
\begin{proposition} \label{thm:diff_semialg}
    Let $K$ be a simplicial complex and $\Phi \colon A \to \bR^{|K|}$ a definable parametrized family of filtrations.
    Then there exists a finite definable partition of $A$, $A = S \sqcup O_1 \sqcup \dots \sqcup O_k$ 
    such that 
    \begin{enumerate}
        \item $\dim S < \dim A := m$;
        \item for any $i=1,\dots,k$, $O_i$ is a definable manifold of dimension $m$ and $\mathrm{Pers} \circ \Phi \colon O_i \to \cD$ is differentiable.
    \end{enumerate}
\end{proposition}
\end{comment}

\subsection{Examples of definable families of filtrations}\label{sec:exs}

%Classical families of filtrations are  semi-algebraic or definable in o-minimal structures as shown in the following examples.

%Below we provide a few examples of common families of filtrations that turn out to be semi-algebraic or definable in a more general o-minimal structure. 

%\begin{example}[Vietoris-Rips filtrations]\label{ex:VR}
\paragraph{Vietoris-Rips filtrations.}
    The family of Vietoris-Rips filtrations built on top  of sets of $n$ points $x_1, \dots, x_n \in \bR^d$ is the semi-algebraic parametrized family of filtrations 
    \begin{align}
        \Phi \colon A = (\bR^d)^n  \to  \bR^{|\Delta_n|} = \bR^{2^n-1},
    \end{align}
    where $\Delta_n$ is the simplicial complex made of all the faces of the $(n-1)$-dimensional simplex, defined, for any $x=(x_1, \dots, x_n) \in A$ and any simplex $\sigma \subseteq \{1, \dots, n\}$, by 
    \begin{align}
        \Phi_\sigma(x) = \max_{i,j \in \sigma} \| x_i - x_j \|.
    \end{align}
    One easily checks that the permutation induced by $\mathrm{Pers}$ is constant on the connected components of the complement of the union of the subspaces $S_{i,j,k,l} = \{(x_1, \dots, x_n) : \| x_i - x_j \| = \| x_k - x_l \| \}$ over all the $4$-tuples $(i,j,k,l)$ such that at least $3$ of the $4$ indices $i,j,k,l$ are distinct. 
    This example naturally extends to Vietoris-Rips-like filtrations in the following way. 
    Let $A \subset \cM_n(\bR)$ be the set of $n \times n$ symmetric matrices with non-negative entries and $0$ on the diagonal.
    This is a semi-algebraic subset of the space of $n$-by-$n$ matrices $\cM_n(\bR) \simeq \bR^{n^2}$, of dimension $m = (n-1)(n-2)/2$. 
    The map $\Phi \colon A \to \bR^{|\Delta_n|} = \bR^{2^n}$ defined by
    $\Phi_\sigma(M) = \max_{i,j \in \sigma} m_{i,j}$ for any $M = (m_{i,j})_{1 \leq i,j \leq n} \in A$, is a semi-algebraic family of filtrations. 
    Note that the set $S$ of \cref{thm:diff_semialg} can be chosen to be the set of matrices with at least $2$ entries that are equal.  
%\end{example}

%\begin{example}[Weighted Rips-filtrations]
%\begin{comment}
\paragraph{Weighted Rips filtrations.}
%Weighted Rips filtrations are a generalization of Vietoris-Rips filtrations where weights are assigned to the vertices of the simplicial complex. 
Given a function $f \colon \bR^d \to \bR$, the family of weighted Rips filtrations $\Phi \colon A = (\bR^d)^n  \to  \bR^{|\Delta_n|} = \bR^{2^n}$ associated to $f$ is defined, for any $x=(x_1, \dots, x_n) \in A$ and any simplex $\sigma \subseteq \{1, \dots, n \}$, by
%\vspace{-2mm}
\begin{itemize}
\item $\Phi_\sigma(x) = 2 f(x_j)$ if $\sigma = [j]$;
\item $\Phi_\sigma(x) = \max(2 f(x_i), 2 f(x_j), \| x_i - x_j \| + f(x_i) + f(x_j))$, if $\sigma = [i,j]$, $i \not = j$;
\item $\Phi_\sigma(x) = \max( \Phi_{[i,j]}(x), i,j \in \sigma)$ if $|\sigma| \geq 3$.
\end{itemize}
%\vspace{-2mm}
%\begin{equation}\label{eq:triangle}
% \Phi_\sigma(x) =
% \begin{cases}
%  2 f(x_j), \sigma = [j] \\
%  \max(2 f(x_i), 2 f(x_j), \| x_i - x_j \| + f(x_i) + \\
%  f(x_j)), \sigma = [i,j], i \not = j; \\
%  \max( \Phi_{[i,j]}(x), i,j \in \sigma),  |\sigma| \geq 3.
% \end{cases}
%\end{equation}

Since Euclidean distances and $\max$ function are semi-algebraic, this family of filtrations is definable as soon as the weight function $f$ is definable.  

This example easily extends to the case where the weight function depends on the set of points $x = ( x_1, \dots, x_n )$: the weight at vertex $y$ is defined by $f(x,y)$ with $f \colon (\bR^d)^n \times \bR^d \to \bR$. 
%Again, the resulting family of filtrations is definable as soon as $f$ is definable. 
A particular example of such a family is given by the so-called DTM filtration \cite{anai2020dtm}, where $f(x,y)$ is the average distance from $y$ to its $k$-nearest neighbors in $x$. 
In this case, $f$ is semi-algebraic, and the family of DTM filtrations is semi-algebraic.   

The o-minimal framework also allows us to consider weight functions involving exponential functions \cite{wilkie1996model}, such as, for instance, kernel-based density estimates with Gaussian kernels.

%\begin{example}[Sublevel sets filtrations]
\paragraph{Sublevel sets filtrations.}
Let $K$ be a simplicial complex with $n$ vertices $v_1, \dots, v_n$. Any real-valued function $f$ defined on the vertices of $K$ can be represented as a vector $(f(v_1), \dots, f(v_n)) \in \bR^n$. The family of sublevel sets filtrations $\Phi \colon A=\bR^n \to \bR^{|K|}$ of functions on the vertices of $K$ is defined by
$\Phi_\sigma(f) = \max_{i \in \sigma} f_i$
for any $f = (f_1, \dots, f_n) \in A$ and any simplex $\sigma \subseteq \{1, \dots, n\}$.
This filtration is also known as the \emph{lower-star filtration} of $f$.
The function $\Phi$ is obviously semi-algebraic, and for \cref{thm:diff_semialg} it is sufficient to choose $S = \bigcup_{1 \leq i < j \leq n} \{ f =  (f_1, \dots, f_n) \in A: f_i = f_j \}$.
%\end{example}

\section{Minimization of functions of persistence}\label{sec:funcs}

Using the same notation as in the previous section, recall that the space of persistence diagrams associated to a filtration of $K$ is identified with $\bR^{|K|} = (\bR^2)^p \times \bR^q$, where each point in the $p$ copies of $\bR^2$ is a point with finite coordinates in the persistence diagram and each coordinate in $\bR^q$ is the $x$-coordinate of a point with infinite persistence.

\begin{definition} \label{def:fn-pers}
A function $E \colon \bR^{|K|} = (\bR^2)^p \times \bR^q \to \bR$ is said to be a \emph{function of persistence} if it is invariant to permutations of the points of the persistence diagram: for any $(p_1, \dots, p_p, e_1, \dots, e_q) \in (\bR^2)^p \times \bR^q$ and any permutations $\alpha, \beta$ of the sets $\{ 1, \dots, p \}$ and $\{ 1, \dots, q \}$, respectively, one has
$$E(p_{\alpha(1)},\dots, p_{\alpha(p)}, e_{\beta(1)}, \dots, e_{\beta(q)}) =  E(p_1,\dots, p_p, e_1, \dots, e_q).$$
\end{definition}

%\begin{remark} \label{rmk:E_Lipschitz}
It follows from this permutation invariance and \cref{lemma:semi-alg-persistence}
that if a function of persistence $E \colon \bR^{2p+q} = \bR^{|K|} \to \bR$ is locally Lipschitz, then the composition $E \circ \mathrm{Pers}$ is also locally Lipschitz.     
%\end{remark}
Moreover, if $E$ is definable in an o-minimal structure, then for any definable parametrized family of filtrations $\Phi \colon A \subseteq \bR^d \to \bR^{|K|}$, the composition $\mathcal{L} = E \circ \mathrm{Pers} \circ \Phi \colon A  \to \bR$ is also definable. 
As a consequence, $\mathcal{L}$ has a well-defined Clarke subdifferential $\partial \mathcal{L}(z):=\mathrm{Conv} \{ \lim_{z_i \to z} \nabla \mathcal{L}(z_i) \; : \; \text{$\mathcal{L}$ is differentiable at $z_i$} \}$, since it is differentiable almost everywhere thanks to \cref{thm:diff_semialg}.

%Building on a recent result from \cite{subgradient}, the aim of this section is to show that stochastic subgradient descent applied to $\mathcal{L}$ converges almost surely to a critical point under mild conditions on $\mathcal{L}$. We show that these conditions are satisfied for a wide variety of examples commonly encountered in Topological Data Analysis. 

%\marc{Related to my remark earlier, we may want to assume that $E$ is invariant to some permutations (and maybe continuous, so we can write the remark about Pers being more-or-less 1-Lipschitz in a convenient place).} \fred{Done : see above}

\subsection{Stochastic gradient descent}

To minimize $\mathcal{L}$, we consider the differential inclusion
\begin{align}
    \frac{dz}{dt} \in - \partial\mathcal{L}(z(t)) \quad \text{\rm for almost every $t$},
\end{align}
whose solutions $z(t)$ are the trajectories of the subgradient of $\mathcal{L}$. 
They can be approximated by the standard stochastic subgradient algorithm given by the iterations of 
\begin{equation} \label{eq:stochastic_gradient}
    x_{k+1} = x_k - \alpha_k ( y_k + \zeta_k),\ y_k \in \partial \mathcal{L}(x_k),
\end{equation}
where the sequence $(\alpha_k)_k$ is the learning rate and $(\zeta_k)_k$ is a sequence of random variables.  
In \cite{subgradient}, the authors prove that under mild technical conditions on these two sequences, the stochastic subgradient algorithm converges almost surely to a critical point of $\mathcal{L}$ as soon as $\mathcal{L}$ is locally Lipschitz. 

More precisely, consider the following assumptions, which correspond to Assumption~C in \cite{subgradient}:
\begin{enumerate}
    \itemsep -0.1cm
    \item for any $k$, $\alpha_k \geq 0$, $\sum_{k=1}^\infty \alpha_k = +\infty$ and, $\sum_{k=1}^\infty \alpha_k^2 < +\infty$;
    \item $\sup_k \| x_k \| < +\infty$, almost surely;
    \item denoting by $\mathcal{F}_k$ the increasing sequence of $\sigma$-algebras $\mathcal{F}_k = \sigma(x_j,y_j,\zeta_j, j < k)$, there exists a function $p \colon \bR^d \to \bR$ which is bounded on bounded sets such that almost surely, for any $k$, 
    \begin{align}
        \mathbb{E}[\zeta_k|\mathcal{F}_k] = 0 \ \ \mbox{\rm and} \ \ 
    \mathbb{E}[\| \zeta_k \|^2 |\mathcal{F}_k] < p(x_k).
    \end{align}
\end{enumerate}

These assumptions are standard and not very restrictive. Assumption~1 depends on the choice of the learning rate by the user and is easily satisfied, e.g., taking $\alpha_k = 1/k$. 
Assumption~2 is usually easy to check for most of the functions $\mathcal{L}$ encountered in practice. 
Assumption~3 is a standard condition, which states that, conditioned upon the past, the variables $\zeta_k$ have zero mean and controlled moments; e.g., this can be achieved by taking a sequence of independent and centered variables with bounded variance that are also independent of the $x_k$'s and $y_k$'s.

Under these assumptions, the following result is an immediate consequence of Corollary~5.9 in \cite{subgradient}. 

\begin{theorem} \label{thm:convergence_guarantees}
Let $K$ be a simplicial complex, $A \subseteq \bR^d$, and $\Phi \colon A \to \bR^{|K|}$ a parametrized family of filtrations of $K$ that is definable in an o-minimal structure. 
Let $E \colon \bR^{|K|} \to \bR$ be a definable function of persistence such that $\mathcal{L} = E \circ \mathrm{Pers} \circ \Phi$ is locally Lipschitz.  
Then, under the above assumptions~1, 2, and 3, almost surely the limit points of the sequence $(x_k)_k$ obtained from the iterations of \cref{eq:stochastic_gradient} are critical points of $\mathcal{L}$ and the sequence $(\mathcal{L}(x_k))_k$ converges. 
\end{theorem}

The above theorem provides explicit conditions ensuring the convergence of stochastic subgradient descent for functions of persistence. 
The main criterion to be checked is the local Lipschitz condition for $\mathcal{L}$. 
From the remark following \cref{def:fn-pers}, it is sufficient to check that $\Phi$ and $E$ are Lipschitz. Regarding $\Phi$, it is obvious for the examples of \cref{sec:exs}. 
This is not the case for some other examples, such as the so-called alpha-complex filtration that can be made locally Lipschitz using a simple technical trick---see Appendix.

\subsection{Examples of definable locally Lipschitz functions of persistence}

%\begin{example}[Total persistence]
\paragraph{Total persistence.}
Let $E$ be the sum of the distances of each point of a persistence diagram with finite coordinates to the diagonal: given a persistence diagram represented as a vector in $\bR^{2p+q}$, $D = ((b_1,d_1), \dots, (b_p,d_p), e_1, \dots, e_q)$, 
\begin{align}
    E(D) = \sum_{i=1}^p |d_i - b_i|.
\end{align}
Then $E$ is obviously semi-algebraic, and thus definable in any o-minimal structure.
It is also Lipschitz.  
%\end{example}

%\begin{example}[Wasserstein and bottleneck distance]\label{ex:Wass}
\paragraph{Wasserstein and bottleneck distance}
Given a persistence diagram $D^*$, the bottleneck distance between the regular part of a diagram $D$ and the regular part of $D^*$ (see \cref{def:regess}) is given by 
\begin{align}
    E(D) = d_B(D_{\reg},D^*_{\reg}) = \min_{m} \ \ \max_{(p,p^*) \in m} ||p - p^*||_\infty,
\end{align}
where, denoting $\Delta = \{ (x,x) : x \in \bR \}$ the diagonal in $\bR^2$, $m$ is a partial matching between $D_{\reg}$ and $D^*_{\reg}$, i.e., a subset of $(D_{\reg} \cup \Delta)  \times (D^*_{\reg} \cup \Delta)$ such that every point of $D_{\reg} \setminus \Delta$ and $D^*_{\reg} \setminus \Delta$, appears exactly once in $m$.
%Defined on the set of filtrations of a given simplicial complex $K$, 
One can easily check that the map $E$ is semi-algebraic, and thus definable in any o-minimal structure. It is also Lipschitz.
This property also extends to the case where the bottleneck distance is replaced by the so-called Wasserstein distance $W_p$ with $p \in \mathbb{N}$ \cite{cohen2010lipschitz}, or its approximation, the Sliced Wasserstein distance~\cite{Carriere2017b}.
Optimization of these functions $E$ and other functions of bottleneck and Wasserstein distances
% between diagrams 
have been used, for example, in shape matching in \cite{poulenard2018topological}.  See also the example on 3D shape in Appendix.
%\end{example}

%\begin{example}[Persistence landscapes \cite{bubenik2015statistical}] %\label{ex:LS}
\paragraph{Persistence landscapes \cite{bubenik2015statistical}}
To any given point $p = (x,y)\in \bR^2$ with $x=\frac{b+d}{2}$ and $y=\frac{d-b}{2}$, associate the function $\Lambda_p \colon \bR \to \bR$ defined by
\begin{equation}\label{eq:triangle}
 \Lambda_p(t) =
% \begin{cases}
%  t-x+y & t \in [x-y, x] \\
%  x+y-t & t \in (x,  x+y] \\
%  0 & \text{otherwise}
% \end{cases}
%
 \begin{cases}
  t-b & (t \in [b, \frac{b+d}{2}]) \\
  d-t & (t \in (\frac{b+d}{2}, d]) \\
  0 & (\text{otherwise}).
 \end{cases}
\end{equation}
Given a persistence diagram $D$, the persistence landscape of $D$ is a summary of the arrangement of the graphs of the functions $\Lambda_p$, $p \in D$:
\begin{align}
\lambda_D(k,t) = \kmax_{p \in D} \Lambda_p(t), \ \ t \in [0,T], k \in \mathbb{Z}^+,
\end{align}
where $\kmax$ is the $k$th largest value in the set, or $0$ when the set contains less than $k$ points. 
Given a positive integer $k$, a finite set $\{t_1, \dots, t_n \} \subset \bR$, and a simplicial complex $K$, the map that associates the vector 
$(\lambda_D(k,t_1), \dots, \lambda_D(k,t_n))$ to each persistence diagram $D$ of a filtration of $K$ is Lipschitz \cite{bubenik2015statistical} and clearly semi-algebraic.

Other classical ways to vectorize persistence diagrams are the \emph{linear representations} \cite{chazal2018density} which are also definable in o-minimal structures, such as, e.g., persistence images~\cite{adams2017persistence}---see Appendix. In \cite{divol2019understanding}, the authors give explicit conditions for such representations to be locally Lipschitz.

\section{Numerical illustrations} \label{sec:implementation}

We showed in \cref{sec:ominim,sec:funcs} that the usual stochastic gradient descent procedure of \cref{eq:stochastic_gradient} enjoys some convergence properties for persistence-based functions. 
This means in particular that the algorithms available in standard libraries such as TensorFlow and PyTorch, which implement stochastic gradient descent (among other optimization methods), can be leveraged and used as is for differentiating persistence diagrams, while still ensuring convergence. 
The purpose of this section is to illustrate that our code, which implements the general gradient defined in \cref{thm:diff_semialg} for persistence-based functions, and which is based on Gudhi\footnote{https://gudhi.inria.fr/} and TensorFlow, can be readily used for studying several different persistence optimization tasks. 
On the way, we also suggest regularization terms that one can add to topological losses in order to avoid unwanted behaviors. 
We only present a few applications due to lack of space, and we refer the interested reader to Appendix for more examples.

\paragraph{Point cloud optimization.}

A toy example in persistence optimization is to modify the positions of the points in a point cloud so that its homology is maximized~\cite{gabrielsson2020topology, gameiro2016continuation}. 
In this experiment, we start with a point cloud $X$ sampled uniformly from the unit square $S=[0,1]^2$, and then optimize the point coordinates so that
%We use Adam optimizer and learning rate of $0.1$ with exponential decay, so that 
%We then minimize 
the loss $\mathcal L(X)=P(X)+T(X)$ is minimized. 
Here $T(X):=-\sum_{p\in D}\|p-\pi_\Delta(p)\|_\infty^2$ is a topological penalty, $D$ is the 1-dimensional persistence diagram associated to the Vietoris-Rips filtration of $X$, $\pi_\Delta$ stands for the projection onto the diagonal $\Delta$, and $P(X):=\sum_{x\in X}d(x,S)$ is a penalty term ensuring that the point coordinates stay in the unit square. 
The topological penalty $T(X)$ was used in~\cite{bruel2019topology}, and 
%Using this loss 
ensures that holes are created within the point cloud so that the cardinality of the persistence diagram $D$ is as large as possible.
However, we point out  that if one uses $T(X)$ alone without the penalty $P(X)$, as in~\cite{bruel2019topology}, then convergence is very difficult to reach since inflating the point cloud with dilations can make the topological penalty $T(X)$ arbitrarily small. 
In contrast, using our second term $P(X)$ in addition to $T(X)$ constrains the points to stay in a fixed region $S$ of the Euclidean plane. 
Another effect of the penalty $P(X)$ is to flatten the boundary of the created holes along the boundary of $S$. 
See \cref{fig:pc} for an illustration. 

\begin{figure}[!ht]
    \centering
    \includegraphics[width=3.75cm]{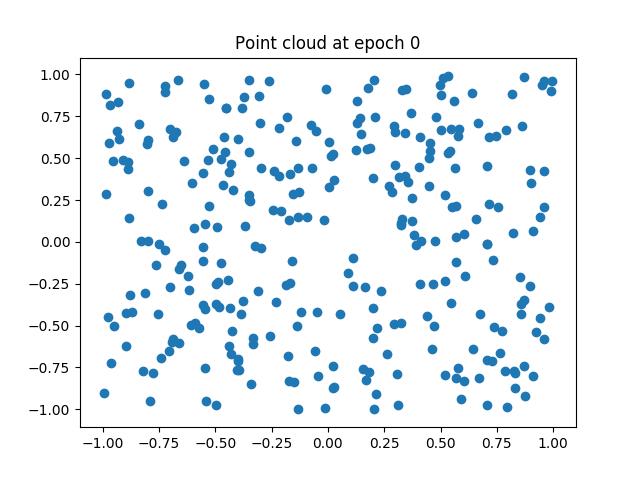}
    \includegraphics[width=3.75cm]{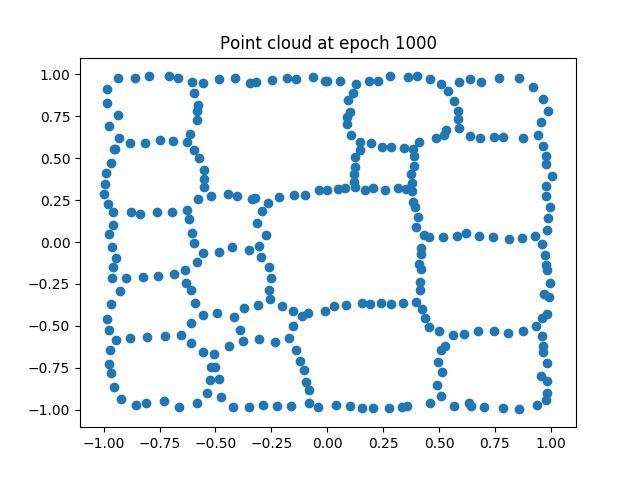}
    \includegraphics[width=3.75cm]{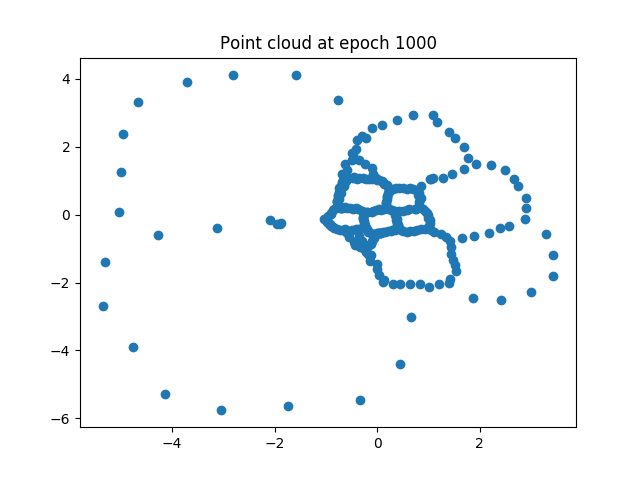}
    \includegraphics[width=3.75cm]{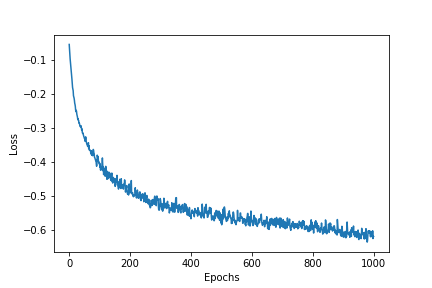}
    \caption{Illustration of point cloud optimization. We initialize with a random point cloud (left), and we show the optimized point cloud when optimization is done with topological and regularization losses (middle left). When only topological loss is used, the optimized point cloud inflated some loops to minimize the loss (middle right). Note how the coordinates are now much larger. We also show the convergence of the total loss (right). }
    \label{fig:pc}
\end{figure}

\paragraph{Dimension reduction.} 

In this experiment, we show how our general setup can be used to reduce dimension with the so-called \emph{topological autoencoders} introduced in~\cite{moor2019topological}. 
In this family of autoencoders, a topological loss $T(X,Z)$ between the input space $X$ and latent space $Z$ is used in addition to the usual loss $\mathcal D(X,Z)=\sum_i \|x_i-z_i\|_2^2$. 
This topological loss was computed in~\cite{moor2019topological} by (i)~computing the permutations induced by the persistence map (see \cref{sec:pers}) of the Vietoris-Rips complexes built from the input space $X$ and the latent space $Z$, (ii)~computing, for each simplex in these permutations, the corresponding edge that induces its filtration value, and (iii)~measuring, for all those edges, 
the differences between the edge lengths in $X$ and the same edge lengths in $Z$. 
To sum up, the loss function is defined as
\begin{equation}
   \mathcal L(X,Z) = \|M_X[\pi_X]-M_Z[\pi_X]\|_2^2 + \|M_X[\pi_Z]-M_Z[\pi_Z]\|_2^2,
\end{equation}
where $M_X,M_Z$ are the distance matrices of the input and the latent spaces respectively, and where $\pi_X,\pi_Z$ denote the indices of the entries in $M_X,M_Z$ that are picked by the permutation induced by the persistence map to generate the Vietoris-Rips persistence diagrams of $X$ and $Z$. 
Note that $\mathcal L$ is obviously semi-algebraic and thus fits in our framework. 
Moreover, in our setup we can directly use the bottleneck and Wasserstein distances between the Vietoris-Rips persistence diagrams of the input and latent spaces as the topological loss. 
%This allows to make the input and latent topological information agree, and not to merely ask the distances that are topologically relevant in both spaces to be preserved, which might be too loose as a constraint.
This is relevant since in \cite{moor2019topological} the authors  pointed out that looking at homology in dimension larger than 1 was not adding anything for their loss, and stuck to 0-dimensional homology. 
We show in \cref{fig:auto} an example in which 1-dimensional homology is also important, that is, a point cloud in $\mathbb R^3$ that is comprised of two nested circles, which is then 
non-linearly embedded in in $\mathbb R^9$ by converting each point $p=(x,y,z)$ into the exponential of the 3x3 anti-symmetric matrix whose coefficient are $x,y$ and $z$. 
%embedded in $\mathbb R^9$ with rotation matrices. 
We then train an autoencoder made of four fully-connected layers with $32$ neurons and ReLU activations, using the usual loss, the usual plus the topological loss described above, and the usual plus a topological loss computed as $\mathcal L(X,Z)=W_1(D_X, D_Z)$, i.e., the 1-Wasserstein distance between the 1-dimensional Vietoris-Rips persistence diagrams of the input and latent spaces. 
%We use Adam optimizer, and learning rate of $0.01$ with exponential decay. 
It can be seen from \cref{fig:auto} that autoencoders without the Wasserstein loss cannot embed the point cloud in the plane perfectly, while using the Wasserstein loss between the 1-dimensional Vietoris-Rips persistence diagrams improves on the result by separating better the two intrinsic circles.

\begin{figure}[!ht]
    \centering
    \includegraphics[width=3.8cm]{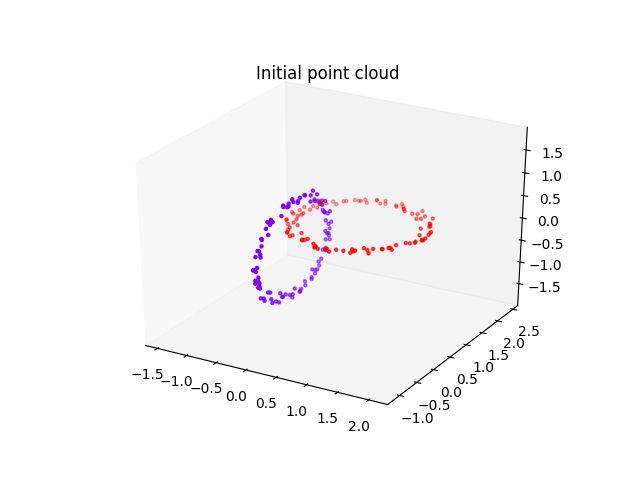}
    \includegraphics[width=3.75cm]{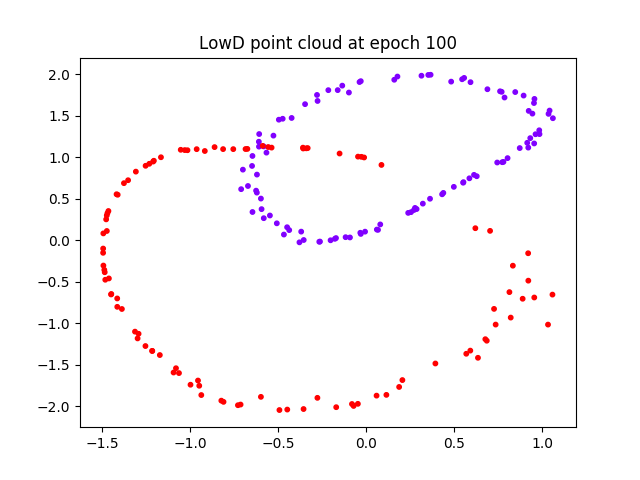}
    \includegraphics[width=3.8cm]{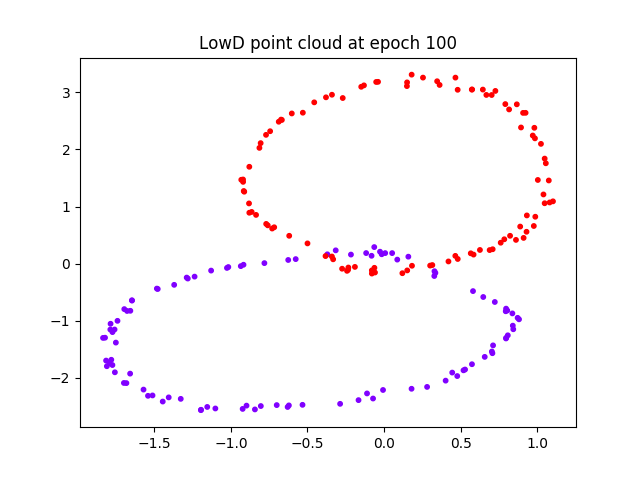}
    \includegraphics[width=3.8cm]{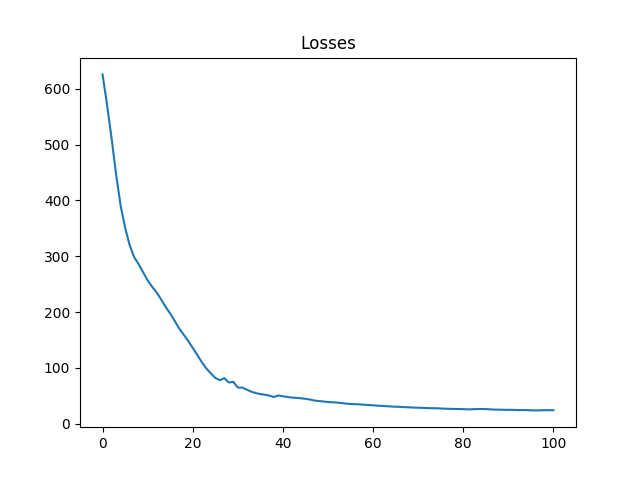}

    \caption{Example of dimension reduction with autoencoders. An initial point cloud made of two circles (left) is embedded in $\mathbb R^9$, and then fed to autoencoders that either do not use topology, or only use the distances induced by the persistence maps in dimension $0$. The resulting embeddings (middle left, we only show one but the two are similar) cannot separate the circles, while using 1-dimensional topology induces a better embedding (middle right). Convergence of the loss function is also provided (right).}
    \label{fig:auto}
\end{figure}

\begin{table*}[!ht]
\begin{minipage}{.5\textwidth}
\raggedleft
\begin{tabular}{c||c||ccc}
Dataset & Baseline & Before & After & Difference \\ 
\hline
\texttt{vs01} & 100.0 & 61.3 & 99.0 & \bf{+37.6} \\ 
\texttt{vs02} & 99.4 & 98.8 & 97.2 & -1.6 \\ 
\texttt{vs06} & 99.4 & 87.3 & 98.2 & \bf{+10.9} \\ 
\texttt{vs09} & 99.4 & 86.8 & 98.3 & \bf{+11.5} \\ 
\texttt{vs16} & 99.7 & 89.0 & 97.3 & \bf{+8.3} \\ 
\texttt{vs19} & 99.6 & 84.8 & 98.0 & \bf{+13.2} \\ 
\texttt{vs24} & 99.4 & 98.7 & 98.7 & 0.0 \\ 
 \texttt{vs25} & 99.4 & 80.6 & 97.2 & \bf{+16.6} \\ 
\end{tabular}
\end{minipage}
\begin{minipage}{.5\textwidth}
\raggedright
\begin{tabular}{c||c||ccc}
Dataset & Baseline & Before & After  & Difference \\ 
\hline
\texttt{vs26} & 99.7 & 98.8 & 98.2 & -0.6 \\ 
 \texttt{vs28} & 99.1 & 96.8 & 96.8 & 0.0 \\ 
 \texttt{vs29} & 99.1 & 91.6 & 98.6 & \bf{+7.0} \\ 
\texttt{vs34} & 99.8 & 99.4 & 99.1 & -0.3 \\ 
 \texttt{vs36} & 99.7 & 99.3 & 99.3 & -0.1 \\ 
 \texttt{vs37} & 98.9 & 94.9 & 97.5 & \bf{+2.6} \\ 
\texttt{vs57} & 99.7 & 90.5 & 97.2 & \bf{+6.7} \\ 
\texttt{vs79} & 99.1 & 85.3 & 96.9 & \bf{+11.5} \\ 
\end{tabular}
\end{minipage}

\caption{\label{tab:scores} Accuracy scores obtained from persistence diagrams before and after performing our optimization over the filtration. Note that the difference between the scores is almost always positive, i.e., there is almost always improvement after our optimization process. Scores %for data sets from \texttt{MNIST} 
do not have standard deviations since we use the train/test splits of the \texttt{mnist.load\_data} function in TensorFlow 2.
%whereas scores for graph data sets are computed with $10$ train/test splits.
}
\end{table*}

\paragraph{Filter selection.} 

In this experiment, we address a very common issue in Topological Data Analysis, filter selection. 
Indeed, when computing persistence diagrams in order to generate topological features from a data set for further data analysis, the filter function that is being used to filter the data set always has to be specified a priori. 
Here, we provide a very simple heuristic to tune it if it comes from a parametrized family $\mathcal F$ of filters and if the learning task is supervised, which is the case in, e.g., classification. We simply start from a random guess in $\mathcal F$ and then optimize the following criterion, which is inspired from \cite{Zhao2019}:
\begin{equation}\label{eq:classloss}
    \mathcal L(f) = \sum_{l=1}^N \frac{\sum_{i,j:y_i=y_j=l} W_p(D_i(f),D_j(f))}{\sum_{i,j:y_i=l} W_p(D_i(f),D_j(f))}, 
\end{equation}
which amounts to minimize the distances between persistence diagrams that share the same label, and increase the distances between persistence diagrams with different labels. 
Note that the batch size that we use in this optimization process has a big influence on the computation time, since the larger the batch size, the more Wasserstein distances we will have to compute in our cost. 
To cope with this issue, we actually used the Sliced Wasserstein distance $SW$~\cite{Carriere2017b} instead of $W_p$, which, since it is computed with projections onto lines, can be defined entirely with matrix operations that are usually available in any library with autodifferentiation. 
This drastically improves on computation time, while remaining in our framework since the Sliced Wasserstein distance is also a semi-algebraic function.

We classify images from the \texttt{MNIST} data set. We assign values to the pixels using a height function given by a direction (parametrized by an angle in the Euclidean plane), and we use 0-dimensional persistence diagrams computed after optimizing this direction using loss~\eqref{eq:classloss}. 
See \cref{fig:mnist}.

\begin{figure}[!ht]
    \centering
    \includegraphics[width=7.5cm]{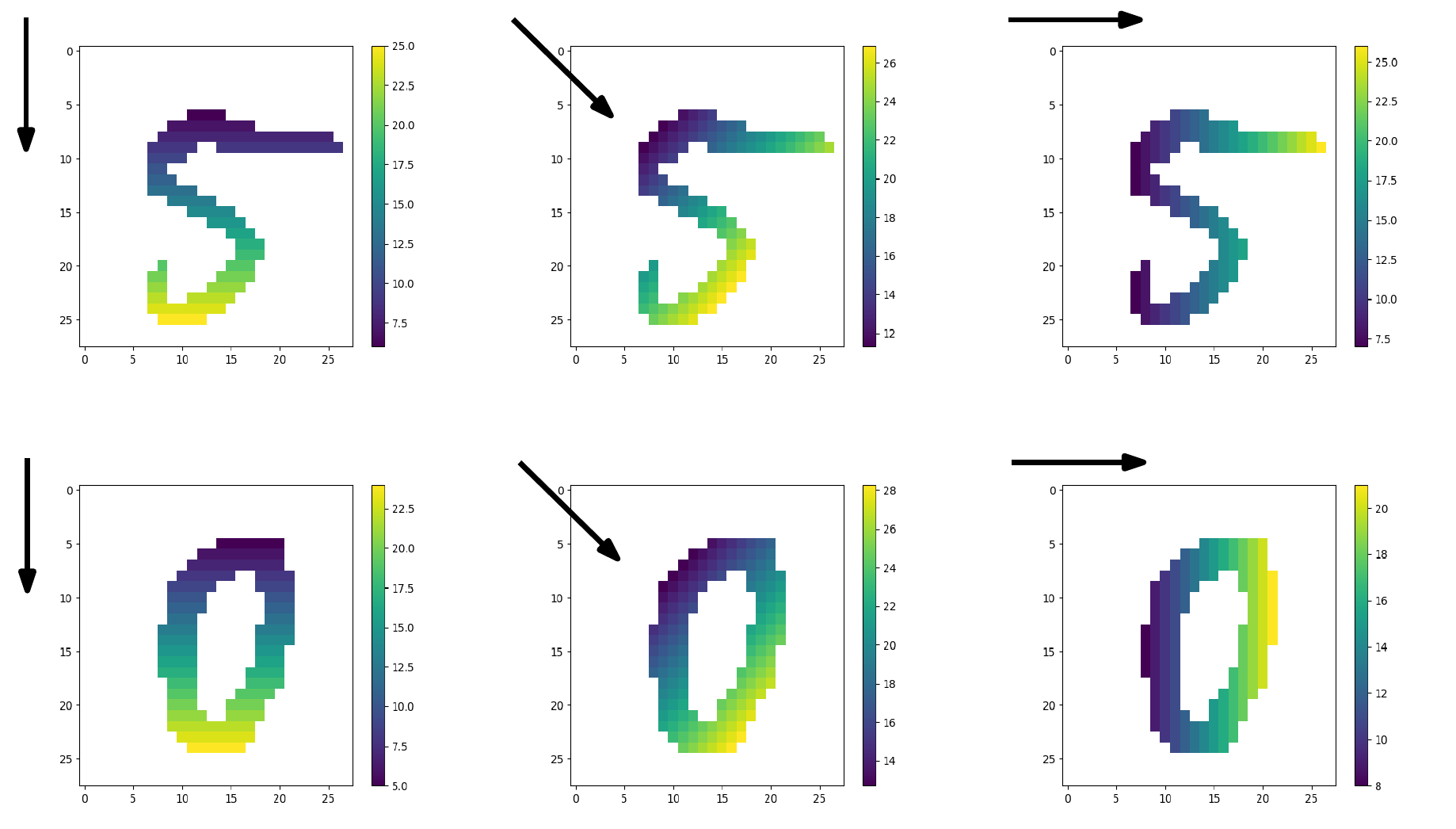}
    \includegraphics[width=7cm]{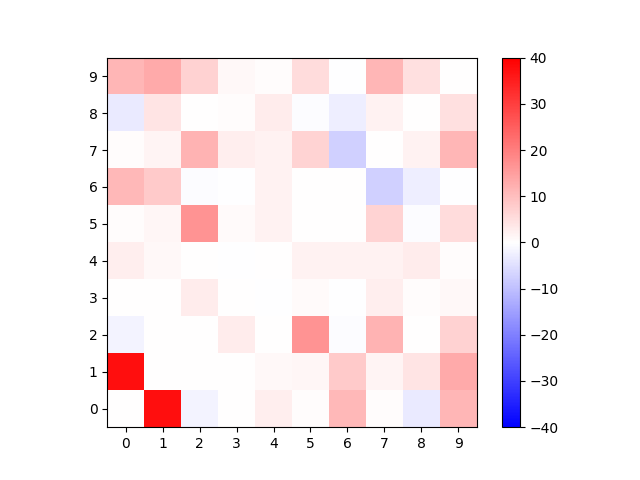}   
    \caption{{\bf Left:} Example of images and directions inducing different height functions. Different directions generate different height functions and filtrations and thus different persistence diagrams. In this experiment, we optimize over the direction so that the persistence diagrams are the most efficient for image classification. {\bf Right:} Pairwise differences between the accuracy scores (\%) of optimized and non optimized filtrations for classifying digit $x$ vs. digit $y$. One can see that the difference is almost always positive.}
    
    \label{fig:mnist}
\end{figure}

We then compute the accuracy scores obtained with a random forest classifier for the (binary) classification tasks digit $x$ vs. digit $y$ for all pairs $0\leq x,y \leq 9$,  using the first five persistence landscapes with resolution $100$ associated to the persistence diagrams before and after optimization. 
Even though our primary goal is to demonstrate that optimizing the filter almost always lead to an improvement, we also add a baseline score obtained by training a random forest classifier directly on the images for proper comparison. 
Some of the scores are displayed in \cref{tab:scores} (the full table can be found in Appendix), and all pairwise scores are shown in \cref{fig:mnist}. 
Interestingly, when starting with a random direction, scores can be much worse than the baseline, but our optimization process is then able to select the best direction that induces the best persistence diagrams (with respect to the classification task) without prior knowledge on the data set.

\section{Conclusion} \label{sec:conclusion}

In this article we introduced a theoretical framework that encompasses most of the previous methods for optimizing topology-based functions. 
In particular, we obtained convergence results for very general classes of functions with topological flavor computed with persistence theory, and provided corresponding code that one can use to reproduce previously introduced topological optimization tasks. 
For future work, we are planning to further investigate tasks related to classifier regularization in ML~\cite{chen2019topological}, and to improve on computation time using, e.g., vineyards~\cite{Cohen-Steiner2006}.

\clearpage

\bibliographystyle{alpha}
\bibliography{references}

\newcommand{\etalchar}[1]{$^{#1}$}
\begin{thebibliography}{BGGSSG20}

\bibitem[ACG{\etalchar{+}}20]{anai2020dtm}
Hirokazu Anai, Fr{\'e}d{\'e}ric Chazal, Marc Glisse, Yuichi Ike, Hiroya
  Inakoshi, Rapha{\"e}l Tinarrage, and Yuhei Umeda.
\newblock {DTM}-based {F}iltrations.
\newblock In {\em Topological Data Analysis}, pages 33--66. Springer, 2020.

\bibitem[AEK{\etalchar{+}}17]{adams2017persistence}
Henry Adams, Tegan Emerson, Michael Kirby, Rachel Neville, Chris Peterson,
  Patrick Shipman, Sofya Chepushtanova, Eric Hanson, Francis Motta, and Lori
  Ziegelmeier.
\newblock Persistence images: A stable vector representation of persistent
  homology.
\newblock {\em The Journal of Machine Learning Research}, 18(1):218--252, 2017.

\bibitem[AGH{\etalchar{+}}09]{attali2009persistence}
Dominique Attali, Marc Glisse, Samuel Hornus, Francis Lazarus, and Dmitriy
  Morozov.
\newblock Persistence-sensitive simplification of functions on surfaces in
  linear time.
\newblock {\em TOPOINVIS}, 9:23--24, 2009.

\bibitem[BCY18]{boissonnat2018geometric}
Jean-Daniel Boissonnat, Fr{\'e}d{\'e}ric Chazal, and Mariette Yvinec.
\newblock {\em Geometric and topological inference}, volume~57.
\newblock Cambridge University Press, 2018.

\bibitem[BE17]{bauer2017morse}
Ulrich Bauer and Herbert Edelsbrunner.
\newblock The morse theory of {\v{c}}ech and delaunay complexes.
\newblock {\em Transactions of the American Mathematical Society},
  369(5):3741--3762, 2017.

\bibitem[BGGSSG20]{bruel2020topology}
Rickard Br{\"u}el-Gabrielsson, Vignesh Ganapathi-Subramanian, Primoz Skraba,
  and Leonidas~J Guibas.
\newblock Topology-aware surface reconstruction for point clouds.
\newblock In {\em Computer Graphics Forum}, volume~39, pages 197--207. Wiley
  Online Library, 2020.

\bibitem[BGND{\etalchar{+}}19]{bruel2019topology}
Rickard Br{\"u}el-Gabrielsson, Bradley~J Nelson, Anjan Dwaraknath, Primoz
  Skraba, Leonidas~J Guibas, and Gunnar Carlsson.
\newblock A topology layer for machine learning.
\newblock {\em arXiv preprint arXiv:1905.12200}, 2019.

\bibitem[BR91]{benedetti1991real}
Riccardo Benedetti and Jean-Jacques Risler.
\newblock {\em Real algebraic and semialgebraic sets}.
\newblock Hermann, 1991.

\bibitem[Bub15]{bubenik2015statistical}
Peter Bubenik.
\newblock Statistical topological data analysis using persistence landscapes.
\newblock {\em The Journal of Machine Learning Research}, 16(1):77--102, 2015.

\bibitem[CBO{\etalchar{+}}20]{clough2020topological}
James Clough, Nicholas Byrne, Ilkay Oksuz, Veronika~A Zimmer, Julia~A Schnabel,
  and Andrew King.
\newblock A topological loss function for deep-learning based image
  segmentation using persistent homology.
\newblock {\em IEEE Transactions on Pattern Analysis and Machine Intelligence},
  2020.

\bibitem[CCI{\etalchar{+}}20]{carriere2020perslay}
Mathieu Carri{\`e}re, Fr{\'e}d{\'e}ric Chazal, Yuichi Ike, Th{\'e}o Lacombe,
  Martin Royer, and Yuhei Umeda.
\newblock Perslay: {A} neural network layer for persistence diagrams and new
  graph topological signatures.
\newblock In {\em International Conference on Artificial Intelligence and
  Statistics}, pages 2786--2796. PMLR, 2020.

\bibitem[CCO17]{Carriere2017b}
Mathieu Carri{\`{e}}re, Marco Cuturi, and Steve Oudot.
\newblock {Sliced Wasserstein kernel for persistence diagrams}.
\newblock In {\em 34th International Conference on Machine Learning (ICML
  2017)}, volume~70, pages 664--673. JMLR.org, 2017.

\bibitem[CD18]{chazal2018density}
Fr{\'e}d{\'e}ric Chazal and Vincent Divol.
\newblock The density of expected persistence diagrams and its kernel based
  estimation.
\newblock In {\em 34th International Symposium on Computational Geometry (SoCG
  2018)}. Schloss Dagstuhl-Leibniz-Zentrum fuer Informatik, 2018.

\bibitem[CdSGO16]{chazal2012structure}
Fr{\'e}d{\'e}ric Chazal, Vin de~Silva, Marc Glisse, and Steve Oudot.
\newblock {\em The structure and stability of persistence modules}.
\newblock SpringerBriefs in Mathematics. Springer, 2016.

\bibitem[CG20]{carlsson2020topological}
Gunnar Carlsson and Rickard~Br{\"u}el Gabrielsson.
\newblock Topological approaches to deep learning.
\newblock In {\em Topological Data Analysis}, pages 119--146. Springer, 2020.

\bibitem[CNBW19]{chen2019topological}
Chao Chen, Xiuyan Ni, Qinxun Bai, and Yusu Wang.
\newblock A topological regularizer for classifiers via persistent homology.
\newblock In {\em The 22nd International Conference on Artificial Intelligence
  and Statistics}, pages 2573--2582, 2019.

\bibitem[Cos00]{coste2000introduction}
Michel Coste.
\newblock {\em An introduction to o-minimal geometry}.
\newblock Istituti editoriali e poligrafici internazionali Pisa, 2000.

\bibitem[CSEH09]{cohen2009extending}
David Cohen-Steiner, Herbert Edelsbrunner, and John Harer.
\newblock Extending persistence using {P}oincar{\'e} and {L}efschetz duality.
\newblock {\em Foundations of Computational Mathematics}, 9(1):79--103, 2009.

\bibitem[CSEHM10]{cohen2010lipschitz}
David Cohen-Steiner, Herbert Edelsbrunner, John Harer, and Yuriy Mileyko.
\newblock Lipschitz functions have {$L_p$}-stable persistence.
\newblock {\em Foundations of computational mathematics}, 10(2):127--139, 2010.

\bibitem[CSEM06]{Cohen-Steiner2006}
David Cohen-Steiner, Herbert Edelsbrunner, and Dmitriy Morozov.
\newblock {Vines and vineyards by updating persistence in linear time}.
\newblock In Nina Amenta and Otfried Cheong, editors, {\em 22nd Annual
  Symposium on Computational Geometry (SoCG 2006)}, pages 119--126. Association
  for Computing Machinery, 2006.

\bibitem[DDKL20]{subgradient}
Damek Davis, Dmitriy Drusvyatskiy, Sham~M. Kakade, and Jason~D. Lee.
\newblock Stochastic subgradient method converges on tame functions.
\newblock {\em Found. Comput. Math.}, 20(1):119--154, 2020.

\bibitem[DL19]{divol2019understanding}
Vincent Divol and Th{\'e}o Lacombe.
\newblock Understanding the topology and the geometry of the persistence
  diagram space via optimal partial transport.
\newblock {\em arXiv preprint arXiv:1901.03048}, 2019.

\bibitem[DUC20]{dindin2020topological}
Meryll Dindin, Yuhei Umeda, and Fr{\'e}d{\'e}ric Chazal.
\newblock Topological data analysis for arrhythmia detection through modular
  neural networks.
\newblock In {\em Canadian Conference on Artificial Intelligence}, pages
  177--188. Springer, 2020.

\bibitem[Ede93]{edelsbrunner1993union}
Herbert Edelsbrunner.
\newblock The union of balls and its dual shape.
\newblock In {\em Proceedings of the ninth annual symposium on Computational
  geometry}, pages 218--231, 1993.

\bibitem[EH10]{edelsbrunner2010computational}
Herbert Edelsbrunner and John Harer.
\newblock {\em Computational topology: an introduction}.
\newblock American Mathematical Soc., 2010.

\bibitem[GC19]{gabrielsson2019exposition}
Rickard~Br{\"u}el Gabrielsson and Gunnar Carlsson.
\newblock Exposition and interpretation of the topology of neural networks.
\newblock In {\em 2019 18th IEEE International Conference On Machine Learning
  And Applications (ICMLA)}, pages 1069--1076. IEEE, 2019.

\bibitem[GHO16]{gameiro2016continuation}
Marcio Gameiro, Yasuaki Hiraoka, and Ippei Obayashi.
\newblock Continuation of point clouds via persistence diagrams.
\newblock {\em Physica D: Nonlinear Phenomena}, 334:118--132, 2016.

\bibitem[GNDS20]{gabrielsson2020topology}
Rickard~Br{\"u}el Gabrielsson, Bradley~J Nelson, Anjan Dwaraknath, and Primoz
  Skraba.
\newblock A topology layer for machine learning.
\newblock In {\em International Conference on Artificial Intelligence and
  Statistics}, pages 1553--1563, 2020.

\bibitem[HKND19]{hofer2019connectivity}
Christoph Hofer, Roland Kwitt, Marc Niethammer, and Mandar Dixit.
\newblock Connectivity-optimized representation learning via persistent
  homology.
\newblock In {\em International Conference on Machine Learning}, pages
  2751--2760. PMLR, 2019.

\bibitem[HKNU17]{hofer2017deep}
Christoph Hofer, Roland Kwitt, Marc Niethammer, and Andreas Uhl.
\newblock Deep learning with topological signatures.
\newblock In {\em Advances in Neural Information Processing Systems}, pages
  1634--1644, 2017.

\bibitem[KKZ{\etalchar{+}}20]{kim2020efficient}
K.~Kim, J.~Kim, M.~Zaheer, J.~Kim, F.~Chazal, and L.~Wasserman.
\newblock Efficient topological layer based on persistent landscapes.
\newblock In {\em NeurIPS 2020 (to appear)}, 2020.

\bibitem[LOT19]{leygonie2019framework}
Jacob Leygonie, Steve Oudot, and Ulrike Tillmann.
\newblock A framework for differential calculus on persistence barcodes.
\newblock {\em arXiv preprint arXiv:1910.00960}, 2019.

\bibitem[MHRB19]{moor2019topological}
Michael Moor, Max Horn, Bastian Rieck, and Karsten Borgwardt.
\newblock Topological autoencoders.
\newblock {\em arXiv preprint arXiv:1906.00722}, 2019.

\bibitem[PSO18]{poulenard2018topological}
Adrien Poulenard, Primoz Skraba, and Maks Ovsjanikov.
\newblock Topological function optimization for continuous shape matching.
\newblock In {\em Computer Graphics Forum}, volume~37, pages 13--25. Wiley
  Online Library, 2018.

\bibitem[RTB{\etalchar{+}}18]{rieck2018neural}
Bastian Rieck, Matteo Togninalli, Christian Bock, Michael Moor, Max Horn,
  Thomas Gumbsch, and Karsten Borgwardt.
\newblock Neural persistence: {A} complexity measure for deep neural networks
  using algebraic topology.
\newblock {\em arXiv preprint arXiv:1812.09764}, 2018.

\bibitem[SWB20]{solomon2020fast}
Elchanan Solomon, Alexander Wagner, and Paul Bendich.
\newblock A fast and robust method for global topological functional
  optimization.
\newblock {\em arXiv preprint arXiv:2009.08496}, 2020.

\bibitem[Ume17]{umeda2017time}
Yuhei Umeda.
\newblock Time series classification via topological data analysis.
\newblock {\em Information and Media Technologies}, 12:228--239, 2017.

\bibitem[Wil96]{wilkie1996model}
Alex~J Wilkie.
\newblock Model completeness results for expansions of the ordered field of
  real numbers by restricted pfaffian functions and the exponential function.
\newblock {\em Journal of the American Mathematical Society}, 9(4):1051--1094,
  1996.

\bibitem[WLSC20]{wangtopogan}
Fan Wang, Huidong Liu, Dimitris Samaras, and Chao Chen.
\newblock Topogan: A topology-aware generative adversarial network.
\newblock In {\em European Conference on Computer Vision(ECCV)}, 2020.

\bibitem[ZW19]{Zhao2019}
Qi~Zhao and Yusu Wang.
\newblock {Learning metrics for persistence-based summaries and applications
  for graph classification}.
\newblock In Hanna Wallach, Hugo Larochelle, Alina Beygelzimer, Florence
  D'Alch{\'{e}}-Buc, Emily Fox, and Roman Garnett, editors, {\em Advances in
  Neural Information Processing Systems 32 (NeurIPS 2019)}, pages 9855--9866.
  Curran Associates, Inc., 2019.

\end{thebibliography}

\newpage

\appendix

\section*{Proof of Proposition 3.4}

{\bf Proposition 3.4.}{\em
Given a simplicial complex $K$, the map $\mathrm{Pers} \colon \mathrm{Filt}_K \subseteq \bR^{|K|} \to  \bR^{|K|}$ is semi-algebraic, and thus definable in any o-minimal structure. Moreover, there exists a semi-algebraic partition of $\mathrm{Filt}_K$ such that the restriction of $\mathrm{Pers}$ to each element of this partition is a Lipschitz map. }

\begin{proof}
As $K$ is finite, the number of preorders on the simplices of $K$ is finite. Let $\preceq$ be a preorder on simplices of $K$ induced by some equalities and inequalities between the coordinates of $\bR^{|K|}$. 
Then, the set of filtrations $F \in \mathrm{Filt}_K$ such that $F$ gives rise to a preorder equal to $\preceq$ is a semi-algebraic set. 
Thus, $\mathrm{Filt}_K$ is a semi-algebraic set that admit a semi-algebraic partition such that the restriction of the persistence map $\mathrm{Pers}$ to each component is a constant permutation. 
As a consequence, on each open element of this partition, the partial derivatives of $\mathrm{Pers}$ are equal constant equal to $0$ or $1$.  

Now, from the stability theorem for persistence \cite{chazal2012structure}, the persistence modules induced by two filtrations $F_1,F_2 \in \mathrm{Filt}_K$ are $\epsilon$-interleaved where $\epsilon$ is the sup norm of the vector $F_2 - F_1$. 
As a consequence, the restriction of $\mathrm{Pers}$ to each component of the above semi-algebraic partition of $\mathrm{Filt}_K$ is $1$-Lipschitz with respect to the sup norm in $\bR^{|K|}$. 
\end{proof}

\section*{Other examples of definable families of filtrations}

\paragraph{The \v Cech and alpha-complex filtrations}
The \v Cech complex filtration built on top  of sets of $n$ points $x_1, \dots, x_n \in \bR^d$ is the semi-algebraic parametrized family of filtrations 
\begin{align}
    \Phi \colon A = (\bR^d)^n  \to  \bR^{|\Delta_n|} = \bR^{2^n-1},
\end{align}
where $\Delta_n$ is the simplicial complex made of all the faces of the $(n-1)$-dimensional simplex, defined, for any $(x_1, \dots, x_n) \in A$ and any simplex $\sigma \subseteq \{1, \dots, n\}$, by 
\begin{align}
    \Phi_\sigma(x) = \min \left\{ r \geq 0 \;:\;  \bigcap_{j \in \sigma} B(x_j,r) \not = \emptyset \right\}
\end{align}
where $B(x_j,r) = \{ x \in \bR^d: \| x - x_j \| \leq r \}$ is the ball of center $x_j$ and radius $r$. The filtrations naturally satisfy the Lipschitz property.
    
%\textcolor{red}{\bf TODO: alpha-complex or just leave the following sentences:}
Note that the \v Cech complex filtration is closely related to the so-called alpha-complex filtration which is a filtration of the Delaunay triangulation of the set of points $x_1, \dots, x_n$---see, e.g., Chapter~6 in \cite{boissonnat2018geometric} for a definition. 
The simplicial complex on which the alpha-complex filtration is built depends on the points $x_1, \dots, x_n$. 
However, if $A$ is a connected component of the open semi-algebraic subset of $(\bR^d)^n$ of the points $x_1, \dots, x_n$ that are in general position\footnote{The points $x_1, \dots, x_n \in \bR^d$ are in general position if any subset of size at most $d+1$ is a set of affinely independent points, and if no subset of $d+2$ points lies on the same $(d-1)$-dimensional sphere}, then all the points in $A$ have the same Delaunay triangulation. 
In that case the alpha-complex defines a semi-algebraic parametrized family of filtrations. 

Moreover, the persistence diagram of the alpha-complex filtration built on top of $x_1, \dots, x_n \in \bR^d$ is the same as the persistence diagram of the \v Cech complex filtration built on the same set of points \cite{edelsbrunner1993union,bauer2017morse} if we ignore points on the diagonal. 
% \marc{Cite https://arxiv.org/abs/1312.1231 instead? A paper from 1993 cannot mention persistence...}

\paragraph{Cubical complexes}
While this presentation was restricted to simplicial complexes for simplicity, the same properties apply for more general complexes. 
Indeed, as long as the boundary maps are well-defined and associate to each cell of dimension $d$ a chain of dimension $d-1$, the same persistence algorithm can be run, and its output is still a permutation of its input. 
% Is this too general? Should we use the words "chain complex"?
The most common non-simplicial complexes are the so-called cubical complexes, where cells are cubes. 
They are particularly well suited to represent images, or (discretized) functions on $\bR^d$.

\section*{Other examples of locally Lipschitz functions of persistence}

\paragraph{Persistence images \cite{adams2017persistence}}
Given a weight function $w \colon \bR^2 \to \bR^{+}$ and a real number $\sigma >0$, the persistence image (also called the persistence surface) associated to a persistence diagram $D$ is the function $\rho_D \colon \bR^2 \to \bR$ defined by
\begin{align}
    \rho_D(q) = \sum_{p \in D} w(p) \exp\left(-\frac{\|q-p\|^2}{2 \sigma^2}\right).
\end{align}
The definition of $\rho_D$ only involves algebraic operations, the weight function $w$, and exponential functions. 
Thus, it follows from \cite{wilkie1996model} that if $w$ is a semi-algebraic function, $K$ is a simplicial complex, and $\{q_1, \dots, q_n \}$ is a finite set of points in $\bR^2$, then the map that associates the vector $(\rho_D(q_1), \dots, \rho_D(q_n))$ to each persistence diagram $D$ of a filtration of $K$ is definable in some o-minimal structure. 
In \cite{adams2017persistence}, the authors provide explicit conditions under which the map from persistence diagrams to persistence surfaces is Lipschitz.

\section*{A remark about the locally Lipschitz condition in Theorem 4.2}

The convergence Theorem 4.2 provides explicit conditions ensuring the convergence of stochastic gradient descent for functions of persistence. 
The main criterion to be checked is the local Lipschitz condition for $\mathcal{L}$. 
From the remark following Definition 4.1, it is sufficient to check that $\Phi$ and $E$ are Lipschitz. Regarding $\Phi$, it is obvious for Vietoris-Rips, \v Cech, DTM, etc.\ filtrations on finite point clouds, but wrong for the alpha-complex filtration, 
%where 3 points with diameter~1 can create an edge and a triangle with an arbitrarily large filtration value. 
%However, such large filtration values only ever participate in persistence intervals of length~0 (points on the diagonal): indeed, the alpha-complex filtration and the \v Cech filtration have the same persistence diagram, if diagonal is ignored. Filtering out 0-length intervals from the alpha-complex thus gives a locally Lipschitz function and can be seen as a technical detail of computing the \v Cech filtration.
%The alpha-complex filtration 
which is not a locally Lipschitz function of the coordinates of the points. 
Indeed, one can take 3 points within a bounded region of the plane that are almost aligned and whose circumradius is arbitrarily large, and this circumradius is the filtration value of the triangle. 
However, by comparing with the \v Cech complex, we know that such large values always correspond to diagonal points of the persistence diagram. 
In our exemple of 3 points, the longest edge has the same filtration as the triangle and they are paired by the persistence algorithm. 
To handle alpha-complex filtrations, we need to restrict to functions of persistence that are defined for various numbers of points, ignore points on the diagonal (the image of a diagram is the same if we add or remove points on the diagonal), and are still locally Lipschitz. 
This is the case for all the functions of persistence presented in this paper. 
Composed with such a function, the difference between \v Cech and alpha-complex filtrations disappears, it becomes an implementation detail and all the differentiation and optimization properties proved for the first apply to the second.

%\marc{redundant with the alpha-complex section above, should we move the last part of the alpha-complex section here?}
%\yuichi{I agree. We should move the part about Lipschitz here and merge it with the above.}

\section*{More numerical experiments}

\paragraph{3D shape processing.} 

In~\cite{poulenard2018topological}, persistence optimization is used for modifying functions defined on 3D shapes. 
More specifically, given a 3D shape $S$, one is interested in optimizing a function $F \colon V(S)\rightarrow \mathbb{R}$ defined on the vertices $V(S)$ of $S$, so that the Wasserstein distance between the persistence diagram associated to $F$ and $D^*$ is minimized, where $D^*$ is a target persistence diagram (which either comes from another function $G \colon S\rightarrow\mathbb{R}$, or is defined by the user). 
In this experiment, we minimize the loss $\mathcal L(F)=T(F)$, where $T(F):=W_2(D,D^*)^2$, that is, the Wasserstein distance between the 0-dimensional persistence diagram $D$ associated to the sublevel set of a function $F$---initialized with the height function, see \cref{fig:sh} (up)---of a human 3D shape, and a target persistence diagram $D^*$ which only contains a single point, with the same coordinates than the point (in the persistence diagram of the height function of the shape $S$) corresponding to the right leg.
This makes the function values to force the two points in $D$ corresponding to the hands of $S$ to go to the diagonal, by creating paths between the hands and the hips (middle). It is worth noting that these path creations come from the fact that we only used a topological penalty in the loss: in~\cite{poulenard2018topological}, the authors ensure smoothness of the resulting function by forcing it to be a linear combination of a first few eigenfunctions of the Laplace-Beltrami operator on the 3D shape. %\mathieu{Not totally sure of what I'm saying here...}
We also display the sequence of optimized persistence diagrams in \cref{fig:sh} (lower row), from which it can be observed that the optimization is piecewise-linear, which reflects the fact that the persistence map has an associated semi-algebraic partition, as per Proposition~3.2. 
%\cref{lemma:semi-alg-persistence}.

\begin{figure}[!ht]
    \centering
    \includegraphics[width=6cm]{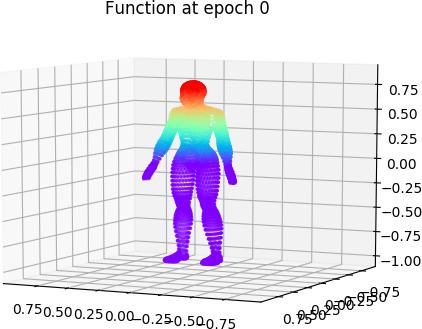}
    \includegraphics[width=6cm]{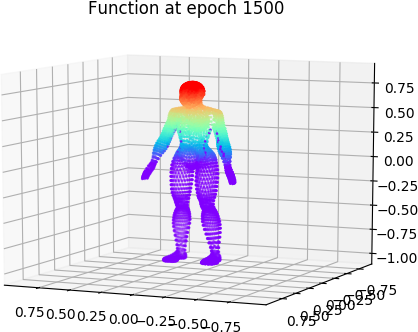}
    \includegraphics[width=5cm]{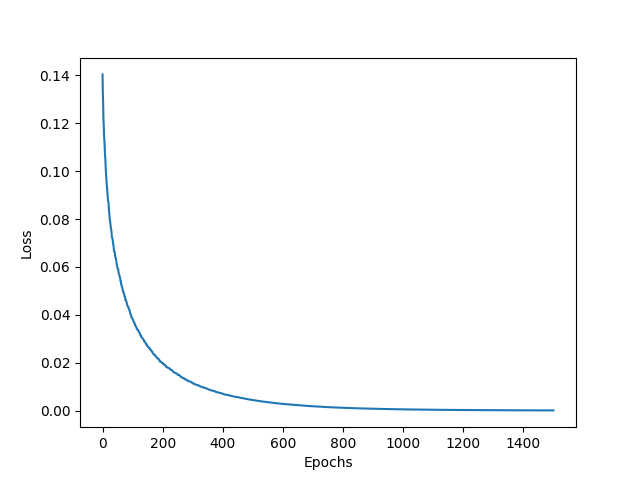}
    \includegraphics[width=5cm]{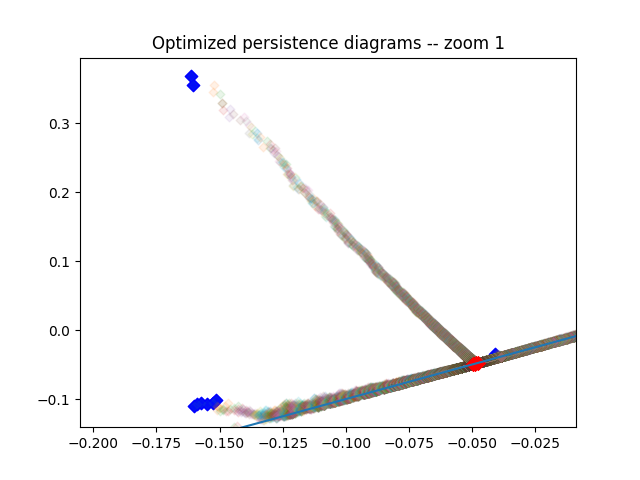}
    \includegraphics[width=5cm]{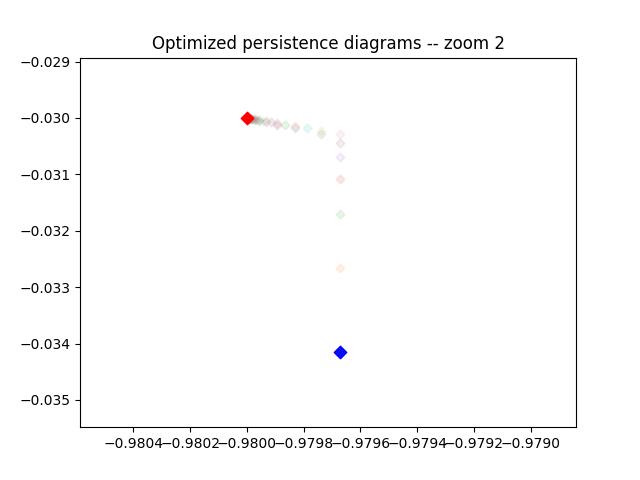}
    \caption{3D shape before (upper left) and after (upper right) optimization, and corresponding loss function (lower left). Note how paths of low function values were created between the hips and the hands.
    We also show the sequence of persistence diagrams (lower middle and right) with blue points being the initial persistence diagram, and red ones being the fully optimized persistence diagram.
    Note how optimization trajectories look piecewise-linear.}
    \label{fig:sh}
\end{figure}

\paragraph{Image processing.} 

Another task presented in~\cite{gabrielsson2020topology} is related to image processing. 
In this experiment, we optimize the 0-dimensional homology associated to the pixel values of a digit binary image $I$ with noise (see \cref{fig:im}, upper left). 
Since noise can be detected as unwanted small connected components, we use the loss $\mathcal L(I)=P(I)+T(I)$, where $T(I):=\sum_{i=1}^p|d_i-b_i|$ is the total persistence penalty, $D_{\reg}=\{(b_1, d_1),\dots, (b_p, d_p)\}$ is the finite 0-dimensional persistence diagram of the cubical complex associated to $I$, and $P(I):=\sum_{p\in I}{\rm min}\{|p|,|1-p|\}$ is a penalty forcing the pixel values to be either $0$ or $1$. 
As can be seen from \cref{fig:im}, using the two penalties is essential: if only $P(I)$ is used, the noise is amplified (lower left), and if only $T(I)$ is used, the noise does not quite disappear, and paths are created between the noise and the central digit to ensure the corresponding connected components are merged right after they appear (lower middle). 
Note that this funny behavior that appears when penalizing topology alone is similar to what was observed in experiments where persistence was used to simplify functions~\cite{attali2009persistence}. 
Using both penalties completely remove the noise (lower right) in two steps: firstly topology is optimized, and then the paths created by optimizing $T(I)$ are removed by optimizing $P(I)$. 
This two step process can also be observed on the loss function (upper middle) and the sequence of optimized persistence diagrams of the image (upper right), where a bifurcation point can be observed in the optimization process. 

%Minimizing this loss amounts to getting rid of the noisy connected components. See \cref{fig:im}, in which the  loss function has two regimes: the persistence diagram is optimized first, leading to branches connecting the noise and the central digit, and then image distance is penalized, leading to the disappearance of those branches.
%\marc{Call it ``total persistence''? The notation $|p|$ doesn't quite match what is used before. It looks like when it sees an extra component, it does a mix of making that component less high, and connecting it to the object. I wonder how one might work around that, maybe penalizing a bit the distance to the original image if we imagine that the connecting path has more pixels than the stain.}
%\fred{This is in the same spirit as the use of total variation (TV) penalty in image denoising. Total $0$-dim persistence can be seen as a variant of TV (this is the same in dimension 1) $\to$ A sentence/ref about that? }

\begin{figure}[!ht]
    \centering
    \includegraphics[width=5cm]{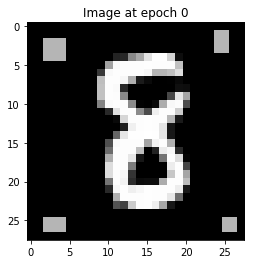}
    \includegraphics[width=5cm]{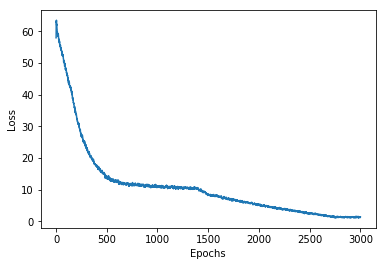}
    \includegraphics[width=5cm]{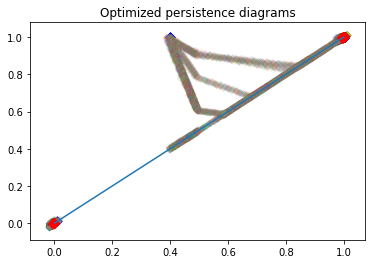}
    \includegraphics[width=5cm]{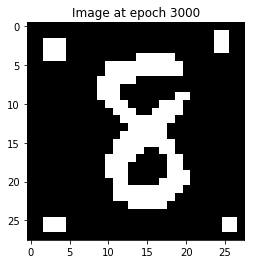}
    \includegraphics[width=5cm]{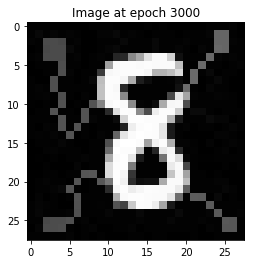}
    \includegraphics[width=5cm]{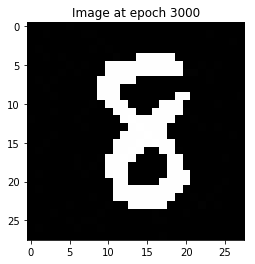}
    
    \caption{Image before (upper left) and after optimization for various penalties. When both $T(I)$ and $P(I)$ are used (lower right), we also show the corresponding loss function (upper middle) and the sequence of persistence diagrams (upper right) with blue points being the initial persistence diagram, and red ones being the fully optimized persistence diagram.}
    \label{fig:im}
\end{figure}

\paragraph{Linear regression.} 

%\begin{wrapfigure}{r}{0.25\textwidth}
%    \centering
%    \includegraphics[width=0.25\textwidth]{images/betastar.png}
%    \label{fig:betastar}
%\end{wrapfigure}

This experiment comes from~\cite{gabrielsson2020topology}, in which the authors use persistence optimization in a linear regression setting. 
Given some dataset $X\in\mathbb{R}^{n\times p}$ and ground-truth associated values $Y\in\mathbb{R}^n$ computed as $Y=X\cdot\beta^* + \epsilon$, where $\beta^*\in\mathbb{R}^p$ is the vector of ground-truth coefficients and $\epsilon$ is some high-dimensional Gaussian noise, one can leverage some prior knowledge on the shape of $\beta$ to penalize the coefficients with bad topology. 
In particular, when using $\beta^*$ with three peaks, as in \cref{fig:re} (left), we use the loss $\mathcal L(\beta)=P(\beta) + TV(\beta) + T(\beta)$, where $T(\beta):=\sum_{i=1}^p|d_i-b_i|$, $\tilde D=\{(b_1,d_1),\dots,(b_p,d_p)\}$ is the 0-dimensional persistence diagram of the sublevel sets of $\beta$, minus the three most prominent points, %(which is some kind of topological equivalent of total variation), 
$TV(\beta)=\sum_i|\beta_{i+1}-\beta_{i}|$ is the usual total variation penalty (which can also be interpreted as a topological penalty as it corresponds to the total persistence of the so-called \emph{extended persistence}~\cite{cohen2009extending} of the signal), and $P(\beta):=\sum_{i} |x_i\cdot\beta-y_i|^2$ is the usual mean-square error (MSE).
We optimized $\beta$ with the MSE alone, then MSE plus total variation, then MSE plus total variation plus topological penalty, and we generated new MSE values from new randomly generated test data, see \cref{fig:re} (right). 
It is interesting to note that using all penalties lead to the best result: using MSE alone leads to overfitting, and adding total variation smooths the coefficients a lot since $\beta$ is initialized with random values. 
Using all penalties ends up being a right combination of minimizing the error, smoothing the signal, and getting to the right shape of $\beta$.
%That is, our loss is the usual MSE used in regression, plus a topological penalty. As observed in~\cite{gabrielsson2020topology}, adding this penalty actually increases the performance of the model on unseen data. See \cref{fig:re}.

\begin{figure}[!ht]
    \centering
    \includegraphics[width=5cm]{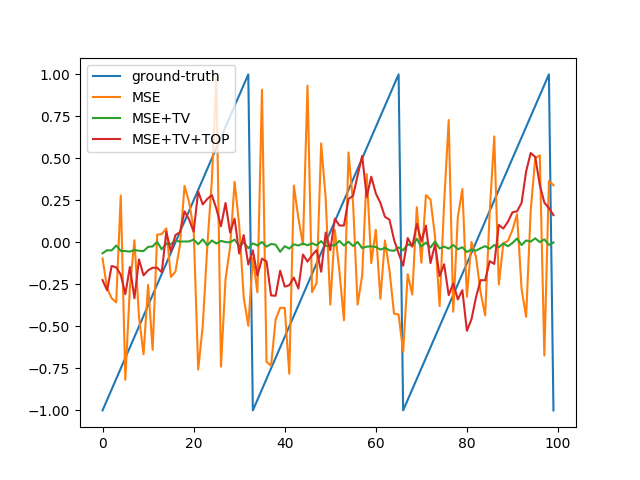}
    \includegraphics[width=5cm]{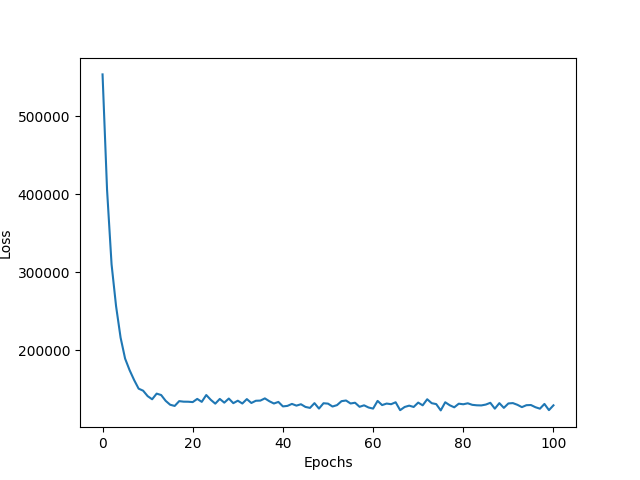}
    \includegraphics[width=5cm]{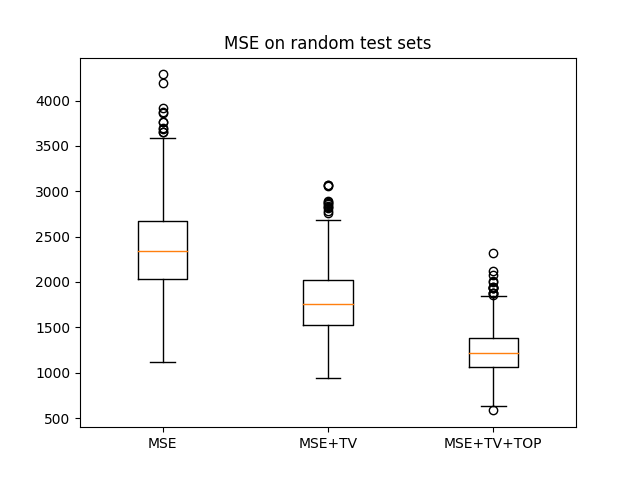}
    \caption{Regression coefficients after optimization for various penalties (left), and corresponding loss function when all penalties are used (middle). Generalization performance is increased by using all penalties, since the MSE on various randomly generated test sets is largely decreased (right).}
    \label{fig:re}
\end{figure}

\paragraph{Noisy point cloud.} 

We perform another point cloud optimization (like in the article), but now we start with a noisy sample $X$ of the circle with three outliers and use the loss $\mathcal L(X)=W_2(D,D^*)^2$, where $W_2$ stands for the Wasserstein distance, $D$ is the 0-dimensional persistence diagram associated to the Vietoris-Rips filtration of $X$, and $D^*$ is the  0-dimensional persistence diagram associated to the Vietoris-Rips filtration of a clean (i.e., with no noise nor outliers) sample of the circle. 
See \cref{fig:noisypc}. 
Note that when one does not use extra penalties, optimizing only topological penalties can lead to funny effects: as one can see on the middle of \cref{fig:noisypc}, the circle got disconnected, and one of the outliers created a small path out of the circle during optimization.

\begin{figure}[!ht]
    \centering
    \includegraphics[width=5cm]{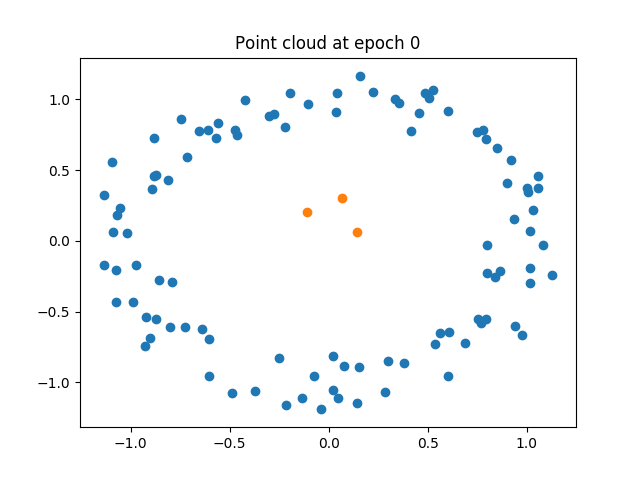}
    \includegraphics[width=5cm]{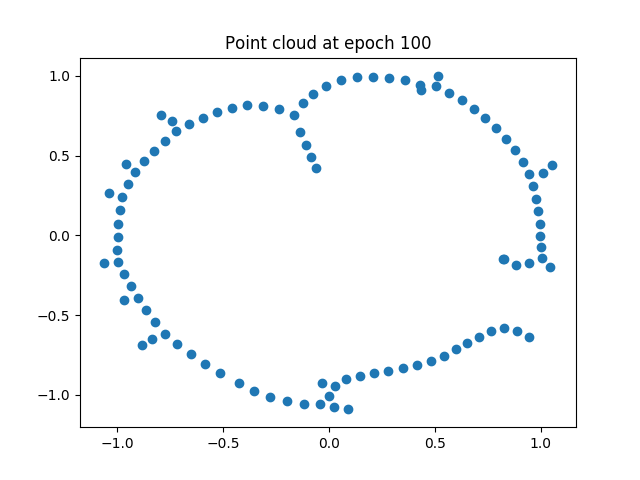}
    \includegraphics[width=5cm]{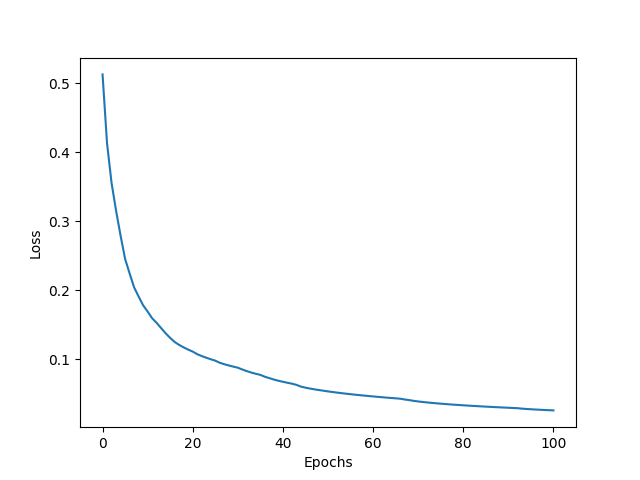}
    \caption{Noisy circle initialization with outliers, before (left) and after (middle) optimization, and corresponding loss function (right).}
    \label{fig:noisypc}
\end{figure}

\paragraph{Filter optimization.}

In addition to classifying digits (as presented in the article), we also run filter selection on a second family of data sets, which is comprised of graphs. For each graph $G$, we compute its normalized Laplacian $L(G)$ and we define a family of filtrations by consider all filter functions that can be written as linear combinations of the eigenvectors of $L(G)$. Then we use persistence landscapes and random forest classifiers (as in the article) to generate accuracy scores. We also train a baseline random forest classifier directly on the eigenvalues of the graph Laplacians. 
The results are displayed (as well as the full results for the \texttt{MNIST} experiment) in \cref{tab:fscores}.

\begin{table*}[!ht]
\centering
\centering
\begin{tabular}{c||c||ccc}
Dataset & Baseline & Before & After & Difference \\ 
\hline
\texttt{all} & 96.9 & 65.1 & 80.6 & \bf{+15.5} \\ 
\texttt{vs01} & 100.0 & 61.3 & 99.0 & \bf{+37.6} \\ 
\texttt{vs02} & 99.4 & 98.8 & 97.2 & -1.6 \\ 
 \texttt{vs03} & 99.8 & 99.1 & 99.2 & \bf{+0.1} \\ 
\texttt{vs04} & 99.9 & 96.0 & 98.8 & \bf{+2.8} \\ 
\texttt{vs05} & 99.6 & 95.7 & 96.3 & \bf{+0.6} \\ 
\texttt{vs06} & 99.4 & 87.3 & 98.2 & \bf{+10.9} \\ 
\texttt{vs07} & 99.8 & 97.4 & 98.0 & \bf{+0.6} \\ 
\texttt{vs08} & 99.4 & 90.4 & 87.0 & -3.4 \\ 
 \texttt{vs09} & 99.4 & 86.8 & 98.3 & \bf{+11.5} \\ 
\texttt{vs12} & 99.6 & 98.3 & 98.5 & \bf{+0.2} \\ 
\texttt{vs13} & 99.7 & 98.9 & 99.1 & \bf{+0.2} \\ 
\texttt{vs14} & 100.0 & 97.1 & 98.3 & \bf{+1.2} \\ 
\texttt{vs15} & 99.8 & 96.7 & 98.0 & \bf{+1.3} \\ 
\texttt{vs16} & 99.7 & 89.0 & 97.3 & \bf{+8.3} \\ 
\texttt{vs17} & 99.7 & 96.8 & 98.6 & \bf{+1.8} \\ 
\texttt{vs18} & 99.8 & 91.7 & 96.0 & \bf{+4.3} \\ 
\texttt{vs19} & 99.6 & 84.8 & 98.0 & \bf{+13.2} \\ 
\texttt{vs23} & 99.1 & 95.2 & 98.0 & \bf{+2.9} \\ 
\texttt{vs24} & 99.4 & 98.7 & 98.7 & 0.0 \\ 
 \texttt{vs25} & 99.4 & 80.6 & 97.2 & \bf{+16.6} \\ 
\texttt{vs26} & 99.7 & 98.8 & 98.2 & -0.6 \\ 
 \texttt{vs27} & 98.6 & 80.1 & 91.9 & \bf{+11.8} \\ 
\texttt{vs28} & 99.1 & 96.8 & 96.8 & 0.0 \\ 
 \texttt{vs29} & 99.1 & 91.6 & 98.6 & \bf{+7.0} \\ 
\texttt{vs34} & 99.8 & 99.4 & 99.1 & -0.3 \\ 
 \texttt{vs35} & 99.2 & 93.5 & 94.3 & \bf{+0.8} \\ 
\texttt{vs36} & 99.7 & 99.3 & 99.3 & -0.1 \\ 
 \texttt{vs37} & 98.9 & 94.9 & 97.5 & \bf{+2.6} \\ 
\texttt{vs38} & 99.0 & 98.3 & 98.8 & \bf{+0.6} \\ 
\texttt{vs39} & 98.8 & 96.8 & 97.8 & \bf{+1.0} \\ 
\texttt{vs45} & 99.9 & 96.5 & 98.4 & \bf{+1.9} \\ 
\texttt{vs46} & 99.6 & 94.1 & 96.0 & \bf{+1.9} \\ 
\texttt{vs47} & 99.7 & 97.2 & 99.3 & \bf{+2.1} \\ 
\texttt{vs48} & 99.2 & 90.4 & 93.4 & \bf{+3.0} \\ 
\texttt{vs49} & 98.4 & 93.7 & 94.2 & \bf{+0.5} \\ 
\texttt{vs56} & 99.0 & 96.9 & 97.1 & \bf{+0.2} \\ 
\texttt{vs57} & 99.7 & 90.5 & 97.2 & \bf{+6.7} \\ 
\texttt{vs58} & 98.9 & 92.7 & 92.3 & -0.4 \\ 
 \texttt{vs59} & 99.4 & 90.0 & 95.4 & \bf{+5.5} \\ 
\texttt{vs67} & 99.7 & 98.4 & 91.0 & -7.4 \\ 
 \texttt{vs68} & 98.7 & 92.2 & 89.5 & -2.7 \\ 
 \texttt{vs69} & 99.7 & 87.0 & 86.7 & -0.3 \\ 
 \texttt{vs78} & 98.9 & 95.7 & 97.6 & \bf{+1.9} \\ 
\texttt{vs79} & 99.1 & 85.3 & 96.9 & \bf{+11.5} \\ 
\texttt{vs89} & 98.7 & 84.2 & 89.1 & \bf{+4.9} \\ 
\texttt{PROTEINS} & 73.6 $\pm$ 3.23 & 68.7 $\pm$ 2.38 & 69.8 $\pm$ 3.42 & \bf{+1.1} \\ 
\texttt{MUTAG} & 85.1 $\pm$ 7.06 & 76.1 $\pm$ 6.31 & 81.3 $\pm$ 6.25 & \bf{+5.2} \\ 
\texttt{COX2} & 78.6 $\pm$ 1.73 & 77.5 $\pm$ 2.29 & 76.9 $\pm$ 3.04 & -0.6 \\ 
 \texttt{DHFR} & 78.8 $\pm$ 4.12 & 61.6 $\pm$ 4.85 & 60.3 $\pm$ 5.26 & -1.3 \\ 
 \texttt{BZR} & 84.9 $\pm$ 2.08 & 78.8 $\pm$ 3.73 & 77.5 $\pm$ 3.93 & -1.2 \\ 
 \texttt{FRANKENSTEIN} & 69.7 $\pm$ 1.40 & 63.3 $\pm$ 2.26 & 63.1 $\pm$ 2.20 & -0.3 \\ 
 \texttt{IMDB-MULTI} & 49.3 $\pm$ 3.26 & 40.6 $\pm$ 2.74 & 39.7 $\pm$ 3.21 & -0.9 \\ 
 \texttt{IMDB-BINARY} & 72.8 $\pm$ 4.45 & 60.1 $\pm$ 3.99 & 60.0 $\pm$ 4.36 & -0.1 \\ 
 \texttt{NCI1} & 74.3 $\pm$ 1.81 & 60.3 $\pm$ 1.27 & 61.4 $\pm$ 1.78 & \bf{+1.0} \\ 
\texttt{NCI109} & 72.5 $\pm$ 1.69 & 60.2 $\pm$ 2.33 & 61.2 $\pm$ 2.87 & \bf{+1.0} \\ 
\end{tabular}

\caption{\label{tab:fscores} All accuracy scores for graphs and \texttt{MNIST} data sets.}
\end{table*}

It should be noted now that even though accuracy scores are sometimes significantly lower than the baseline, optimizing the filter prior to classification often improves the results, and if the score does not improve (which can be due to the fact that the batch size, taken as 150 so that computation runs in a reasonable amount of time, is too small to properly ensure convergence), it does not decrease by a significant margin either.

\section*{Details on experiments}

All experiments were run on Intel dual-Xeon SP with 10 cores and 9.6GB/core (RAM). 
The Classification results for the experiments on graphs were generated with 10-fold cross-validation: each data set was divided into ten 90-10 train-test splits, and results were averaged over these splits. 
The Classification results for the experiments on images were generated using train-test splits provided in \texttt{Tensorflow}.
We also added a multi-class classification task, called {\tt all}, which consists in jointly classifying all images (and not only digit $x$ vs. digit $y$). 
All the optimization processes (including those presented in the article) were done with stochastic gradient descent as implemented in \texttt{Tensorflow} 2.4.1. 
More specifically, we used the \texttt{SGD} optimizer class with \texttt{InverseTimeDecay} learning rate (in order to satisfy Assumptions~1, 2 and 3 of Section~4.1 in the article). We also used \texttt{Adam} optimizer with \texttt{ExponentialDecay} learning rate for some experiments since we noticed empirically that, even though Assumption 1 was not satisfied, the results were not very different, and convergence was slightly smoother. 
Parameter initialization was done randomly, and the batch sizes + numbers of epochs were taken sufficiently large so that convergence was reached in each illustration---see \url{https://github.com/MathieuCarriere/difftda.git} for exact values.  For instance, filter selection was done with initial learning rate equal to $0.001$ and batch size equal to $150$.

\end{document}